\documentclass[11pt]{article}
\usepackage{complexity}

\usepackage[english]{babel}
\usepackage[utf8x]{inputenc}
\usepackage[T1]{fontenc}
\usepackage{enumitem}
\usepackage[margin=1.in,marginparwidth=1.75cm]{geometry}

\usepackage[algo2e, ruled, vlined]{algorithm2e}
\usepackage{algorithm}
\usepackage{amsmath}
\usepackage{amsfonts}
\usepackage{graphicx}
\usepackage{amsthm}
\usepackage{amssymb}
\usepackage{dsfont}
\usepackage[noend]{algorithmic} 
\usepackage{bbm}

\usepackage[T1]{fontenc}
\usepackage{microtype}          %
\usepackage[full]{textcomp}

\usepackage{lipsum}
\usepackage{xspace}
\usepackage{relsize}
\usepackage{framed}
\usepackage{xcolor}
\usepackage{graphicx}
\usepackage{multirow}
\usepackage{subcaption}
\usepackage{float}

\usepackage[colorinlistoftodos]{todonotes}
\usepackage[colorlinks=true, allcolors=blue]{hyperref}
\usepackage{cleveref}

\newtheorem{theorem}{Theorem}

\newtheorem{lemma}{Lemma}
\newtheorem{corollary}{Corollary}
\newtheorem{remark}{Remark}

\theoremstyle{remark}

\theoremstyle{definition}
\newtheorem{definition}{Definition}

\def\final{1}  %
\def\iflong{\iffalse}
\ifnum\final=0  %
\newcommand{\knote}[1]{{\color{blue}[{Karthik: \bf #1}]\marginpar{\color{blue}*}}}
\newcommand{\znote}[1]{{\color{purple}[{Weihao: \bf #1}]\marginpar{\color{purple}*}}}
\newcommand{\cnote}[1]{{\color{magenta}[{Chandra: \bf #1}]\marginpar{\color{magenta}*}}}
\else %
\newcommand{\knote}[1]{}
\newcommand{\znote}[1]{}
\newcommand{\cnote}[1]{}
\fi  %

\usepackage{float}
\usepackage{booktabs}
\usepackage{multirow}
\usepackage{thmtools} 
\usepackage{thm-restate}
\usepackage{lipsum}

\newcommand{\Z}{\mathbb{Z}}
\renewcommand{\R}{\mathbb{R}}

\newif\ifdraft
\draftfalse

\newcommand{\Comp}{\#\text{Comps}}

\newcommand{\SymDif}{\mathop{\triangle}}

\title{Random Deletion for Circuits in Regular Matroids}

\begin{document}
\author{Karthekeyan Chandrasekaran\thanks{Grainger College of Engineering, University of Illinois, Urbana-Champaign. Email: {\tt\{karthe, chekuri, weihaoz3\}@illinois.edu}. Supported in part by NSF grant CCF-2402667.}
\hspace{1.2em} Chandra Chekuri\footnotemark[1]
\hspace{1.2em} Weihao Zhu\footnotemark[1]
}
\date{}
\maketitle

\pagenumbering{gobble}
\begin{abstract}
  We study a simple random deletion algorithm to find short circuits in regular matroids. Given a regular matroid $M$ with $m$ elements, the algorithm randomly deletes each element independently with probability $1-1/m^{\Theta(1/g)}$ where $g$ is the size of a minimum circuit, and returns the remaining set $R$ of elements. We establish the following results:
    
    \begin{itemize}
        \item With high probability, $R$ contains no circuit and is therefore independent.
        \item For every $\alpha$-approximate minimum circuit $C$ where $\alpha\geq 1$, with probability at least $m^{-O(\alpha)}$, the entire circuit $C$ survives in $R$ and moreover, the minimum number of elements that need to be removed to make $R$ independent is $O(\alpha\log(\alpha+2))$. This yields an algorithmic proof that the number of $\alpha$-approximate minimum circuits is
        \[
            m^{O(\alpha)},
        \]
        improving the previous known bound of $m^{O(\alpha^2)}$ by Gurjar and Vishnoi~\cite{gurjar2019number}, and resolving a recent conjecture of Khanna, Putterman, and Song \cite{khanna2025optimal}. 
    \end{itemize}
    Our results yield two immediate applications. First, we obtain very simple RNC algorithms for computing the shortest circuit, and for enumerating all $\alpha$-approximate circuits, of a regular matroid without computing the Seymour decomposition. Second, we obtain a randomized parallel algorithm to find a basis of a regular matroid under independence-oracle access. It uses polynomially many independence-oracle queries in $O(\log m)$ adaptive rounds and succeeds with high probability. This follows from the framework of Khanna, Putterman, and Song~\cite{khanna2025optimal} and the improved counting bound, and extends their
    result from graphic matroids to regular matroids.
  \end{abstract}

\newpage
\pagenumbering{arabic}
\section{Introduction}\label{section:introduction}
Karger's celebrated random contraction algorithm~\cite{karger1993global,karger1995random} for the global minimum cut problem in graphs has been highly influential. In addition to giving a completely new algorithmic approach to the minimum cut problem, it also established a bound on the number of near-minimum cuts in undirected graphs. Let $G$ be an edge-weighted graph on $n$ vertices and let $\lambda$ denote its minimum-cut value. The algorithm iteratively selects an edge with probability proportional to its weight and contracts it.  Karger showed that every fixed cut of weight at most $\alpha\lambda$ survives with probability $n^{-O(\alpha)}$ if the contraction process is stopped when the graph has $\Theta(\alpha)$ super-vertices. At this point, all cuts of the contracted graph can be enumerated since there are only $2^{O(\alpha)}$ such cuts.  Combining these two allows one to conclude that the number of $\alpha$-approximate cuts is $n^{O(\alpha)}$. It also yields a randomized algorithm to enumerate all such cuts. Subsequently there have been other approaches that proved tighter bounds  and also led to deterministic enumeration algorithms \cite{karger2000minimum, NNI97, CQX20, BCW23}.

Similar counting bounds are known for approximate shortest cycles in a graph. Consider a graph with $m$ edges and positive edge weights/lengths\footnote{Zero length edges lead to exponential bounds for trivial reasons and hence we disallow them.}.  Teo and Koh~\cite{teo1992number} proved that an unweighted graph has at most $m^2$ shortest cycles. Subramanian~\cite{subramanian1995polynomial} independently showed that the number of $\alpha$-approximate shortest cycles in a \emph{weighted} graph with positive edge-weights is at most $(2m)^{\lceil 2\alpha \rceil}$, and that these cycles can be enumerated in $m^{O(\alpha)}\poly(n,m)$ time. Subramanian was explicitly motivated by Karger's result and by the connection between cycles in graphs and cuts in graphs via a matroidal viewpoint that we will explain shortly.  Nevertheless, the proofs in \cite{teo1992number} and \cite{subramanian1995polynomial} follow a very different approach from that of Karger, and are based on deterministic combinatorial arguments.  Subramanian's short paper uses the notion of a robot exploring the graph in a breadth-first search fashion.

We now explain the connection between cuts and cycles in graphs from a matroidal point of view. We assume some familiarity with matroids (see Section~\ref{section:preliminaries} for formal definitions). Let $M=(E, \mathcal{I})$ be a matroid. A circuit of $M$ is an inclusion-minimal dependent set and a central computational problem of interest in matroids is that of finding a smallest circuit.\footnote{Smallest circuit of a matroid is also known as the girth of the matroid, but we reserve the term girth for the shortest cycle of a graph in this work to avoid confusion.}   
Given an undirected multigraph $G=(V,E)$, the graphic matroid $M_G=(E,\mathcal{I})$ associated with $G$ is $(E, \mathcal{I})$ where $F \in \mathcal{I}$ iff $F$ is a forest in $G$. The circuits of a graphic matroid are precisely the simple cycles of the underlying graph. Each matroid $M$ has a dual-matroid $M^*$ defined on the same ground set. The dual of the graphic matroid $M_G$ is the cographic matroid $M^*_G$. The circuits of $M^*_G$ are the minimal edge cuts (bonds) in $G$. 
Thus, minimum circuit size in $M_G$ is the girth of $G$, and the minimum circuit size in $M^*_G$ is the global min-cut of $G$. 
In terms of matroid duality, \emph{deleting} an edge in the graphic matroid $M_G$ corresponds to \emph{contracting} an edge in the co-graphic matroid $M^*_G$.

As we remarked earlier, Subramanian explicitly motivated his approximate-cycle counting bound via the preceding connection between $M_G$ and $M^*_G$, and Karger's results on approximate mincuts. However, it appears that the natural adaptation of Karger's random contraction algorithm to find (approximate) shortest cycles in graphs had not been considered in the literature. The correspondence between deletion and contraction operations suggests a randomized algorithm for girth: iteratively pick an edge in proportion to its weight and delete it until arriving at a graph with a unique cycle. In recent work on parallel basis finding, Khanna, Putterman and Song \cite{khanna2025optimal} considered a related algorithm that deletes edges in a one-shot manner. This and prior structural reasons motivated us to consider the random deletion algorithm for counting approximate short cycles in graphs. The cyclomatic number of a graph is the minimum number of edges that must be removed to make it acyclic. We showed very recently \cite{chandrasekaran2026randomDeletion} that for every fixed $\alpha \ge 1$, an $\alpha$-approximate shortest cycle survives while the cyclomatic number of the remaining graph decreases to $O(\alpha^2)$ with probability $m^{-O(\alpha)}$. Since a graph of cyclomatic number $d$ contains at most $2^d-1$ cycles, this gives an algorithmic proof that the number of $\alpha$-approximate cycles in a graph is $2^{O(\alpha^2)}m^{O(\alpha)}$.
Note that the bound we obtain is slightly weaker than the one given by
Subramanian. Nevertheless, this is a different perspective that led to 
our work in this paper on regular matroids that we describe next.
To explain our results and how our cycle result generalizes, we need a natural generalization of the cyclomatic number in graphs to matroids. In a matroid $(E,\mathcal{I})$ for every set $R\subseteq E$, we define its \emph{nullity} as $\nu_M(R):=|R|-r_M(R)$ 
where $r_M(R):=\max\{|I|: I\subseteq R, I\in \mathcal{I}\}$ is the rank of $R$. Equivalently, $\nu_M(R)$ is the minimum number of elements that need to be removed from $R$ to make it independent. In a graphic matroid, nullity coincides with the cyclomatic number of the corresponding graph.

Regular matroids are an important class of matroids that were instrumental in the development of matroid theory via the work of Tutte, who characterized them via excluded minors as well as established their connection to totally unimodular matrices \cite{Tut58} (see also \cite{cunningham2012coming}). Regular matroids are representable over every field, and in particular, they are a subclass of binary matroids and are minor-closed. Seymour gave a powerful algorithmic result on recognizing regular matroids via his fundamental decomposition theorem \cite{seymour1980decomposition}. His theorem shows that every regular matroid can be constructed from graphic, cographic, and a fixed ten-element matroid $R_{10}$ using $1$-, $2$-, and $3$-sums. Such a decomposition can be computed in polynomial time \cite{seymour1980decomposition,truemper-book} and it can further be used to compute a smallest circuit of a regular matroid in polynomial time (given access to the matroid via its independence-oracle) \cite{kapadia-blogpost}.  In contrast, finding a smallest circuit of binary matroids is NP-Hard \cite{Var97}.

The preceding connections therefore suggest two broad questions:
\begin{enumerate}
    \item Can one derive a counting bound for near-minimum circuits in regular matroids comparable to the known bounds for cuts and cycles in graphs?
    \item Does the natural random deletion algorithm have comparably strong performance on regular matroids in terms of finding shortest circuits and approximately-short circuits?
    \end{enumerate}

The first question has received substantial attention. Let $M$ be a regular matroid on $m$ elements, and let $g$ be its minimum circuit size. Gurjar, Thierauf, and Vishnoi~\cite{gurjar2018isolating} first proved that $M$ has $O(m^5)$ circuits of size strictly less than $\frac{3}{2}g$. In particular, this gives a polynomial bound for every $1\leq \alpha < \frac{3}{2}$. Gurjar and Vishnoi~\cite{gurjar2019number} later proved that, for every $\alpha\geq 1$, the number of circuits of size at most $\alpha g$ is bounded by $m^{O(\alpha^2)}$. This established polynomial growth for every fixed $\alpha$, but left a quadratic dependence on $\alpha$ in the exponent of $m$, compared with the linear dependence known for graphic and cographic matroids. A more recent motivation to examine these questions comes from parallel basis finding of regular matroids. Khanna, Putterman, and Song ~\cite{khanna2025optimal} pointed out that if one can prove a uniform bound of $m^{O(\alpha)}$ on the number of $\alpha$-approximate minimum circuits then the machinery that they developed for  graphic matroids would go through to yield a fast parallel algorithm to find a basis via independence-queries.

In this work, we address the preceding questions through the following one-shot random deletion algorithm. For a regular matroid $M=(E,\mathcal{I})$ with $m$ elements, let $g$ be its minimum circuit size. The algorithm independently deletes each element with probability $1-m^{-40/g}$ and returns the remaining/residual set $R$; when $g$ is large this probability is roughly $\frac{40 \ln m}{g}$. The pseudocode is given in \Cref{algo:random-delete}.

\begin{algorithm}[ht]
\caption{Random-Delete\((M=(E, \mathcal{I}), g\in \Z_+)\)}\label{algo:random-delete}
\begin{enumerate}[leftmargin=2em]
\item $p \gets 1 - |E|^{-40/g}$.  
\item For every $e\in E$, delete it with probability $p$.
\item $R \gets $ the set of remaining elements.
\item Return $R$.
\end{enumerate}
\end{algorithm}

We establish two properties of \Cref{algo:random-delete}. First, the remaining set $R$ is independent with high probability. Second, for every fixed circuit $C$ with $|C|\leq \alpha g$, the circuit $C$ survives in $R$ with probability at least $m^{-O(\alpha)}$, and conditioned on its survival, the nullity of the returned set $R$ is $O\left(\alpha\log(\alpha+2)\right)$ with high probability. We may assume $\alpha\leq m$, since otherwise every circuit is an $\alpha$-approximate minimum circuit. This tradeoff yields a new algorithmic proof that the number of $\alpha$-approximate minimum circuits is at most $m^{O(\alpha)}$, as well as a corresponding randomized enumeration algorithm, improving the previous $m^{O(\alpha^2)}$ bound of Gurjar and Vishnoi~\cite{gurjar2019number}. An $m^{O(\alpha)}$ bound was recently conjectured by Khanna, Putterman, and Song \cite{khanna2025optimal} and hence our results resolve this conjecture in the affirmative.\footnote{Although Khanna, Putterman, and Song \cite{khanna2025optimal} stated this conjecture explicitly, it is a natural question after the work of Gurjar and Vishnoi \cite{gurjar2019number}.}

\subsection{Our Results}
We now formalize our results. Our main result is that the random deletion algorithm returns an independent set with high probability.

\begin{restatable}{theorem}{theoremtargetfree}\label{theorem:target-free}
  Let $M=(E,\mathcal{I})$ be a regular matroid with $m\geq 2$ elements and minimum circuit size $g$. \Cref{algo:random-delete} returns a set $R$ that is independent with probability at least $(1-1/m^5)$.
\end{restatable}

We recall that the nullity of a set is the minimum number of elements whose removal makes it independent. A key fact that follows from the preceding theorem is that the nullity of the remaining set, conditioned on the survival of a fixed near-minimum circuit, can be bounded. 

\begin{restatable}{theorem}{theoremnullity}\label{theorem:nullity}
    Let $M=(E,\mathcal{I})$ be a regular matroid with $m\geq 2$ elements and minimum circuit size $g$. We fix a circuit $C$ of size $|C|\leq \alpha g$ with $\alpha\geq 1$. The set $R$ returned by \Cref{algo:random-delete} satisfies
    \begin{align*}
        \Pr [C\subseteq R] &\geq m^{-O(\alpha)}\text{ and}\\
        \Pr\left[ \nu_M(R)= O\left(\alpha\log(\alpha+2)\right) \middle| C\subseteq R\right] &\geq 1-\frac{1}{m^2}.
    \end{align*}
\end{restatable}

We observe that regular matroids are binary, so a set of nullity $b$ contains at most $2^b$ circuits (see Lemma \ref{lemma:binary-matroid-nullity}). Combining these facts with \Cref{theorem:nullity} yields the following counting bound.

\begin{corollary}\label{cor:counting}
    Let $M=(E,\mathcal{I})$ be a regular matroid with $m\geq 2$ elements and minimum circuit size $g$. For every $\alpha\geq 1$, the number of circuits with size at most $\alpha g$ is at most
    \[
        2^{O\left(\alpha\log(\alpha+2)\right)}m^{O(\alpha)}=(m\alpha)^{O(\alpha)}=m^{O(\alpha)}.
    \]
\end{corollary}

\Cref{algo:random-delete} further implies an enumeration algorithm for near-minimum circuits. By repeating the random deletion process and enumerating all circuits in each resulting set of bounded nullity, we obtain the following corollary.

\begin{corollary}\label{cor:enumeration}
    Let $M=(E,\mathcal{I})$ be a regular matroid with $m\geq 2$ elements and minimum circuit size $g$. For every $\alpha\geq 1$, there is a randomized algorithm that accesses the matroid via independence oracle queries and enumerates all circuits of size at most $\alpha g$ in
    \[
        2^{O\left(\alpha\log(\alpha+2)\right)}m^{O(\alpha)}=(m\alpha)^{O(\alpha)}=m^{O(\alpha)}
    \]
    time with high probability.
\end{corollary}

\begin{remark}
  Our results in this paper improve and subsume our previous unpublished work \cite{chandrasekaran2026randomDeletion}. We need the ideas from that work since graphic matroids serve as a base case in our proof for regular matroids via Seymour's decomposition. We folded the ideas from that unpublished work, along with improvements, into this paper.
\end{remark}

A counting bound of $m^{O(\alpha)}$ on the number of $\alpha$-approximate shortest circuits easily implies that random deletion returns an independent set with high probability via a standard union bound (i.e., Corollary \ref{cor:counting} implies Theorem \ref{theorem:target-free}). This is also implicit in \cite{khanna2025optimal}. Our work goes the other way---we use the fact that random deletion returns an independent set with high probability to infer the $m^{O(\alpha)}$ bound on the number of $\alpha$-approximate shortest circuits (i.e., use Theorem \ref{theorem:target-free} to infer Corollary \ref{cor:counting}). Thus,  the two statements are equivalent up to constants in the $O(\alpha)$ term in the counting bound. Theorem \ref{theorem:nullity} is the bridge to infer the counting bound from Theorem \ref{theorem:target-free}.

\medskip
We outline a few applications of our main results. 

\paragraph{Parallel Algorithm for Smallest Circuit in Regular Matroids.} Theorem~\ref{theorem:nullity} yields a simple randomized parallel algorithm for computing a smallest circuit of a regular matroid in the independence oracle model: the random deletion process is non-adaptive and can be done using $m^{O(1)}$ processors in one round. First, we assume that the min-circuit size $g$ is known.
By Theorem \ref{theorem:nullity}, a fixed min-circuit survives in the returned set $R$ and the nullity of $R$ is at most a constant with probability at least $m^{-O(1)}$. Suppose both events happen. Then, all circuits in $R$ can be enumerated in constant time using $m^{O(1)}$ processors and $m^{O(1)}$ independence queries in one round. Thus, the random deletion process can be used to find
a circuit of size at most $g$ using $m^{O(1)}$ processors and $m^{O(1)}$ independence queries in one round with probability $m^{-O(1)}$. Repeating this $m^{O(1)}\log{m}$ times independently (using as many processors) boosts the success probability to $1-m^{-\Omega(1)}$. We can dispense with the assumption that $g$ is known by trying all possible values of $g$ in parallel and returning the smallest circuit found.

As we remarked earlier, a sequential algorithm via Seymour's decomposition is known \cite{kapadia-blogpost}. In contrast, the random deletion algorithm is extremely simple and does not require computing the decomposition. As far as we are aware, this is the first known parallel algorithm for this problem. We note that even for computing the girth of a graph, the random deletion algorithm provides an alternate parallel algorithm (standard algorithms via BFS and matrix multiplication and other methods are known for this already). Our algorithms are, in a sense, the analogues of Karger's random contraction algorithm that gave the first RNC algorithms for global mincut. 

\paragraph{Parallel basis finding.}
We consider the problem of finding a basis in a matroid using independence queries. How fast can we find a basis in parallel? We refer the reader to recent exciting work of Khanna, Putterman and Song \cite{khanna2025optimal} for motivation and past work on this problem. In that work, they established the optimal $O(\log m)$ round complexity for graphic matroids. They asked whether one can obtain a similar result for regular matroids. We affirmatively answer their question. Our results imply that there exists a randomized parallel algorithm that finds a basis of a regular matroid in $O(\log{m})$ rounds using $\poly(m)$ non-adaptive queries to the independence oracle in each round with high probability. 
This result follows from the framework in \cite{khanna2025optimal} in a black-box fashion once we have the $m^{O(\alpha)}$ bound on $\alpha$-approximate circuits. 
\cite{khanna2025optimal} obtain a deterministic algorithm for graphic matroids. We do not attempt to derandomize our algorithm in this version.

\paragraph{Regular Codes and TU Lattices.} Gurjar and Vishnoi \cite{gurjar2019number} used the bound on the number of $\alpha$-approximate shortest circuits in regular matroids to bound the number of $\alpha$-approximate codewords in regular codes and the number of $\alpha$-approximate shortest vectors in TU lattices. Our improved bound on the number of $\alpha$-approximate shortest circuits in regular matroids directly improves the bounds for these applications as well. We refer the reader to Gurjar-Vishnoi's paper to derive these improvements \cite{gurjar2019number}. 

\paragraph{Extensions and discussion.} Our main contribution is conceptual. We investigated random deletion as a unified way to compute and bound approximate shortest circuits in regular matroids after realizing that it works for graphic matroids. This generalizes Karger's random contraction algorithm for finding a global min-cut in graphs. The algorithm 
generalizes to any matroid and we believe that it has the desired properties for several generalizations of regular matroids. We discuss some of these classes. MFMC matroids are a generalization of regular matroids defined by Seymour \cite{Seymour1977}. The decomposition theory of regular matroids extends to these \cite{seymour1980decomposition,truemper-book}. Gurjar and Vishnoi \cite{gurjar2019number} extended their counting bound to this class and we believe that our approach also extends. Recall that 
computing the minimum circuit size of binary matroids is NP-Hard. 
 Geelen, Gerards and Whittle conjectured \cite{Geelenetal2015} that minimum-circuit problem is polynomial-time solvable for every \emph{proper} minor-closed families of binary matroids. They characterized a subfamily of these matroids using a notion called $t$-perturbation. Geelen and Kapadia \cite{geelen2018computing}, motivated by this conjecture, considered the problem of minimum circuit size in $t$-perturbations of graphic and co-graphic matroids and obtained randomized polynomial-time algorithms for them. They use different approaches for these two cases and build on sophisticated tools from combinatorial optimization as well as Karger's random contraction algorithm for the cographic case. One can show that the random deletion approach works in a unified fashion, not only for those two cases, but for $t$-perturbations of regular matroids as well. It is natural to ask if the random deletion would work for every proper minor-closed family of binary matroids; if it does, then this would answer the conjecture from \cite{Geelenetal2015} affirmatively (modulo the randomization). We leave this to future work.

\subsection{Overview of Techniques and AI Disclosure}\label{sec:techniques}
The impetus for our work came initially from \cite{khanna2025optimal} who showed that random deletion of edges in a graph yields a forest with high probability. To prove this fact they use the counting bound of Subramanian \cite{subramanian1995polynomial}. The starting point of our work was to ask whether we can prove Subramanian's counting bound via random deletion rather than the other way around. In our preliminary work, we achieved this by focusing on the nullity of the remaining set after random deletion. We proved a weaker version of Theorem~\ref{theorem:nullity} as follows. Fix an $\alpha$-approximate shortest cycle $C$. We define a \emph{$C$-witness} to be an inclusion-wise minimal path of length at least $g(G)/4$ that is edge-disjoint from $C$. Consider the graph $G_T$ returned by the deletion algorithm. We show that, with probability at least $m^{-O(\alpha)}$, every edge of $C$ survives in $G_T$ \emph{and} $G_T$ contains no $C$-witness. We then prove that the absence of such a witness implies that $\nu(G_T)=O(\alpha^2)$. This yields a counting bound of $2^{O(\alpha^2)}m^{O(\alpha)}$ which is weaker than $m^{O(\alpha)}$ for large $\alpha$. More recently we improved the nullity bound to $O(\alpha\log(\alpha+2))$ and realized that this tighter bound can be obtained by combining Theorem~\ref{theorem:target-free} and structural results on girth.

The cographic case can be handled by Karger's results on mincuts. With the graphic and cographic counting bounds obtained by the same algorithm, it is natural to consider the random deletion algorithm for regular matroids via Seymour's decomposition theorem~\cite{seymour1980decomposition}. It states that every regular matroid can be constructed via $1$-, $2$-, or $3$-sums of two smaller regular matroids. A first idea is to apply this decomposition recursively and combine the guarantees obtained for the two smaller matroids. However, a circuit can pass through both smaller matroids without appearing as a circuit in either one after the common elements are removed. Consequently, the desired guarantees cannot be applied recursively in a direct way. We take inspiration from prior work of Gurjar and Vishnoi~\cite{gurjar2019number} who developed several technical tools in their proof that showed a bound of $m^{O(\alpha^2)}$. The main difference is that we are analyzing a very different algorithm based on random deletion while they relied on Subramanian's work for the graphic matroid case. Our preliminary work used a two step argument where we first proved Theorem~\ref{theorem:target-free} for regular matroids, and then used it along with a more complex witness based argument to prove that nullity of the remaining set after random deletion is at most $O(\alpha^{1+\epsilon})$. Subsequently, as with the setting of cycles in graphs, we realized that one can use structural results on regular matroids to prove Theorem~\ref{theorem:nullity} via Theorem~\ref{theorem:target-free}. Although it is easy to obtain a weaker bound of $O(\alpha \log m)$ on the nullity from Theorem~\ref{theorem:target-free}, the bound of $O(\alpha\log(\alpha+2))$ is based on a stronger relationship between the girth and nullity for regular matroids. Such a bound follows from a recent work of Knauer \cite{Kn26} and we also give an alternative proof of this bound via Theorem~\ref{theorem:target-free}.

\paragraph{Clean unordered decomposition tree.} Gurjar and Vishnoi~\cite{gurjar2019number} organized the basic matroids arising from Seymour's decomposition into an unordered decomposition tree $\mathcal{T}$ (they borrow ideas from Chun and Oxley \cite{chun2011unavoidable}). Each node $u$ of the tree is graphic, cographic, or $R_{10}$, and deleting an arbitrary tree edge divides the tree into two parts whose associated matroids combine by the corresponding sum operation. In particular each connected component of the tree defines a regular matroid. In this work, we use a similar unordered decomposition tree that additionally satisfies a ``clean'' property. The clean property is important for our work: at every cographic node, the three elements of each $3$-sum correspond exactly to the three edges incident to a degree-three vertex in the underlying graph. This property is essential for the absorption step below.\footnote{Gurjar and Vishnoi's paper \cite{gurjar2019number} does not utilize this additional property. 
It appears that their proof has a technical gap due to this omission. We communicated this to Gurjar who suggested Delta-Wye transformation. Clean unordered decomposition tree along with Delta-Wye transformation can be used to fix the gap in their proof; the authors of \cite{gurjar2019number} plan to make a revised version available.} We infer the existence of a clean unordered decomposition tree using the results of Aprile and Fiorini \cite{aprile2022regular}. 

\paragraph{Matroid absorption.}
Let $\mathcal{T}$ be the decomposition tree and let $u$ be a node of the tree and let $M_u$ be the matroid associated with $u$. Consider an edge $e$ of the tree incident to $u$. Removing $e$ creates a component that attaches to $u$ through a $2$- or $3$-sum. The projection of a circuit $C$ of the matroid onto $M_u$ may contain virtual elements introduced by the decomposition, making the graphic and cographic results difficult to apply directly. Our matroid absorption lemma replaces the attached component with weighted elements so that these projections correspond to circuits in a modified \emph{weighted} matroid. The weights record the elements represented inside the absorbed component and preserve their behavior under random deletion. For a $3$-sum, the modification uses a $\Delta$ - $Y$ transformation; the clean-decomposition property ensures that the resulting matroid remains graphic, cographic, or $R_{10}$. This idea builds on the absorption arguments of Gurjar, Thierauf, and Vishnoi~\cite{gurjar2018isolating} and Gurjar and Vishnoi~\cite{gurjar2019number}. 

\paragraph{Random Deletion returns an Independent Set.} To prove \Cref{theorem:target-free}, we root the clean decomposition tree and associate a good event with each node, chosen so that together these events rule out every surviving circuit. Suppose that the good events hold for all children of a node. The matroid absorption lemma then replaces the child components with weighted gadgets, reducing the analysis at the node to a single modified weighted basic matroid. We construct certificates such that the failure of the node's good event forces one of these certificates to survive. We then prove that, with high probability, no certificate survives in the whole tree, which further shows that the remaining set is independent. These certificates are derived from properties of the basic matroids, and here is where 
our witness argument for cycles under random deletion plays an important role. 

\paragraph{AI Disclosure.} The authors initially studied the problem for graphic matroids, developed the random deletion perspective, and proved a weaker result that every fixed $\alpha$-approximate shortest cycle survives with probability $m^{-O(\alpha^3)}$, and conditioning on its survival, the remaining graph has cyclomatic number $O(\alpha^3)$ with high probability. After the authors provided this result and its proof to GPT-5.6 Sol Ultra, the model produced the stronger bound of $O(\alpha^{2})$ on the cyclomatic number, using techniques similar to the authors' original arguments, and further improved it to $O(\alpha^{1.0001})$. The authors subsequently investigated random deletion algorithm for regular matroids and used GPT-5.6 Sol Ultra to help understand the work of Gurjar and Vishnoi~\cite{gurjar2019number}. During this process, the model identified a technical bug in that work. Gurjar suggested the Delta-Wye transformation as a potential fix and building on this idea and the clean unordered decomposition, the authors derived the improved bound presented in this work with the help of AI. The authors also used GPT-6 Astra to assist with writing this manuscript. The authors assume full responsibility for all content.

\paragraph{Other related work.}
Karger's random contraction algorithm and related ideas, discovered more than 30 years ago, have led to many fundamental developments in graph algorithms including flows and cuts, network (un)reliability, sparsification, and spanning trees to name a few. The ideas have been extended beyond graphs to matroid, hypergraphs, and beyond. There are too many papers to discuss and we will not attempt to survey them. We discuss one thread that is relevant to our focus in this paper. Soon after his work on graphs, Karger pointed out the connection of random contraction and sampling to matroids \cite{karger1995random,Kar-matroid} and derived several applications. As we mentioned, random contraction later became a key ingredient in the randomized polynomial-time algorithms of Geelen and Kapadia~\cite{geelen2018computing} for computing the girth and cogirth of bounded-rank perturbations of graphic matroids.  They used Karger's random contraction algorithm for the cographic case but used a very different technique for the graphic case.
Much more recently, Quanrud~\cite{quanrud2024quotient} introduced quotient sparsification for monotone submodular functions, showing that a weighted ground set can be sparsified while approximately preserving every quotient; this simultaneously recovers graph and hypergraph cut sparsification and yields sparsifiers for matroids. We~\cite{chandrasekaran2026polymatroidal} recently identified quotient-bounded classes of polymatroids for which variants of random contraction yield exact, approximation, and fixed-parameter algorithms for minimum quotients, together with corresponding counting bounds. These works all point to the utility of the (poly)-matroidal view of Karger's ideas. Despite these connections to matroids and Subramanian's motivation, it is only recently that \cite{khanna2025optimal}, motivated by parallel basis algorithms, considered the random deletion algorithm for the graphic case, and they based their analysis on using Subramanian's counting bound as a black-box rather than deriving a counting bound via random deletion. 

\subsection{Paper Organization}

In \Cref{section:preliminaries}, we introduce basic concepts in matroid theory and introduce Seymour's decomposition theorem for regular matroids. \Cref{section:decomposition} develops the clean unordered decomposition tree that underlies our proofs, and \Cref{section:absorption} formalizes matroid absorption. We finish the proof of \Cref{theorem:target-free} in \Cref{section:target-free}. We prove the conditional nullity bound and use it to establish the counting and enumeration corollaries in \Cref{section:new:conditional-nullity}.

\section{Preliminaries}\label{section:preliminaries}

\subsection{Basic Matroid Terminology}\label{sec:prelims-matroids}

\begin{definition}[Matroid]
    A \emph{matroid} is a pair $M=(E,\mathcal{I})$, where $E$ is a finite ground set and $\mathcal{I}\subseteq 2^E$ is a family of independent sets satisfying the following properties:
    \begin{enumerate}
        \item $\emptyset\in \mathcal{I}$;
        \item every subset of an independent set is independent;
        \item and, for every $I, J\in \mathcal{I}$ with $|I|<|J|$, there exists an element $e\in J\setminus I$ such that $I\cup\{e\} \in \mathcal{I}$.
    \end{enumerate}
\end{definition}

A set $S\subseteq E$ is called \emph{dependent} if $S\notin \mathcal{I}$. A \emph{basis} is an inclusion-wise maximal independent set, and a \emph{circuit} is an inclusion-wise minimal dependent set. The minimum circuit size of $M$ is
\[
    g(M):=\min\{|C|: \text{$C$ is a circuit of $M$}\},
\]
and we write $g=g(M)$ when the matroid is clear from context. More generally, if $M$ is equipped with a positive weight function $w:E\rightarrow \R_+$, we write $w(S):=\sum_{e\in S}w(e)$ for every $S\subseteq E$ and define the weighted minimum circuit size by
\[
    g_w(M):=\min\{w(C): \text{$C$ is a circuit of $M$}\}.
\]

\begin{definition}[Rank and Nullity]
    For every subset $S\subseteq E$, the \emph{rank} and \emph{nullity} of $S$ are defined as
    \[
        r_M(S):=\max\{|I|:I\subseteq S, I \in \mathcal{I}\}
        \qquad \text{and} \qquad
        \nu_M(S):=|S|-r_M(S).
    \]
\end{definition}

Equivalently, $\nu_M(S)$ is the minimum number of elements that must be removed from $S$ to make it independent. The rank function of a matroid is a submodular function. We define $r_M:=r_M(E)$ as the \emph{rank} of the matroid $M$.

\begin{definition}[Deletion]\label{definition:deletion}
    For every $A\subseteq E$, the \emph{deletion} of $A$ from $M$, denoted by $M\setminus A$, is the matroid on the ground set $E\setminus A$ whose independent sets are $\{I\in \mathcal{I}: I\subseteq E\setminus A\}$. Its rank function satisfies
    \[
        r_{M\setminus A}(S):=r_M(S) \qquad \text{for every $S\subseteq E\setminus A$}. 
    \]
\end{definition}

The circuits of $M\setminus A$ are those circuits of $M$ that do not contain any element in $A$.

\begin{definition}[Contraction]\label{definition:contraction}
    The \emph{contraction} of $A$ in $M$, denoted by $M/A$, is the matroid on the ground set $E\setminus A$ whose rank function is given by
    \[
        r_{M/A}(S):=r_M(A\cup S)-r_M(A) \quad \text{for every $S\subseteq E\setminus A$}.
    \]
\end{definition}

\begin{definition}[Dual Matroid]
    The \emph{dual} of $M$, denoted by $M^*$, is the matroid on the ground set $E$ whose bases are the complements of the bases of $M$. Its rank function is given by
    \[
        r^*_M(S):=|S|-r_M(E)+r_M(E\setminus S) \qquad \text{for every $S\subseteq E$}.
    \]
\end{definition}
A circuit of $M^*$ is called a \emph{cocircuit} of $M$. We also consider two standard forms of oracle access to a matroid $M=(E,\mathcal{I})$. Given a subset $S\subseteq E$ of the ground set, a \emph{rank oracle} returns $r_M(S)$, whereas an \emph{independence oracle} determines whether $S\in \mathcal{I}$. An independence query can be answered using a single rank query, since $S$ is independent if and only if $r_M(S)=|S|$.

\subsection{Graphic and Cographic Matroids}

Let $G=(V,E)$ be an undirected graph. For every subset $F\subseteq E$ of edges, let $\Comp(V,F)$ be the number of connected components in the subgraph $(V,F)$.

\begin{definition}[Graphic Matroid]
    Let $G=(V,E)$ be an undirected graph. The \emph{graphic matroid} of $G$, denoted by $M(G)$, is the matroid with ground set $E$ in which $F\subseteq E$ is independent if and only if the subgraph $(V,F)$ contains no cycle. Its rank function is
    \[
        r_{M(G)}(F)=|V|-\Comp(V,F) \qquad \text{for every $F\subseteq E$}.
    \]
    A matroid is \emph{graphic} if it is isomorphic to $M(G)$ for some graph $G$.
\end{definition}

For every $F\subseteq E$, $F$ is a circuit of $M(G)$ if and only if $F$ is the edge set of a cycle in $G$. The following theorem states a bound on the number of near-shortest cycles in an undirected graph.

\begin{theorem}[Subramanian~\cite{subramanian1995polynomial}]\label{theorem:counting-cycle}
    Let $G=(V,E)$ be an undirected graph with $m:=|E|$ edges and positive edge weights $w:E\rightarrow \R_+$. Let
    \[
        g:=\min\{w(C):\text{$C$ is a cycle of $G$}\}.
    \]
    Then, for every $\alpha\geq 1$, the number of cycles $C$ satisfying $w(C)\leq \alpha g$ is at most $(2m)^{\lceil 2\alpha \rceil}$.
\end{theorem}

For every $S\subseteq V$, let $\delta_G(S):=\{\{u,v\}\in E: u\in S, v\in V\setminus S \}$ denote the cut induced by $S$. A \emph{bond} of $G$ is an inclusion-wise minimal nonempty cut.

\begin{definition}[Cographic Matroid]
    Let $G=(V,E)$ be an undirected graph. The \emph{cographic matroid} of $G$, denoted by $M^*(G)$, is the dual of its graphic matroid
    \[
        M^*(G):=(M(G))^*.
    \]
    A matroid is \emph{cographic} if it is isomorphic to $M^*(G)$ for some graph $G$.
\end{definition}

For every $F\subseteq E$, $F$ is a circuit of $M^*(G)$ if and only if $F$ is the edge set of a bond in $G$. The following theorem states a bound on the number of near-minimum cuts in an undirected graph.

\begin{theorem}[Karger~\cite{karger1993global,karger1995random}]\label{theorem:counting-cut}
    Let $G=(V,E)$ be a connected undirected graph with $n:=|V|\geq 2$ vertices and nonnegative edge weights $w:E\rightarrow \R_{\geq 0}$. Suppose that its minimum cut weight
    \[
        \lambda:=\min \{w(\delta_G(S)): \emptyset \subsetneq S\subsetneq V\}
    \]
    is positive. Then, for every $\alpha\geq 1$, the number of cuts $\delta_G(S)$ with weight at most $w(\delta_G(S))\leq \alpha\lambda$ is at most $n^{2\alpha}$.
\end{theorem}

\subsection{Binary and Regular Matroids}

\begin{definition}[Binary and Regular Matroids]
    Let $M=(E,\mathcal{I})$ be a matroid. A matrix $A$ over a field $\mathbb{F}$, whose columns are indexed by $E$, \emph{represents} $M$ if $S\subseteq E$ is independent in $M$ if and only if the columns indexed by $S$ are linearly independent over $\mathbb{F}$. The matroid $M$ is \emph{binary} if it is representable over $\mathbb{F}_2$, and \emph{regular} if it is representable over every field.
\end{definition}

Regular matroids are binary, and graphic and cographic matroids are regular. Moreover, the classes of graphic, cographic, and regular matroids are closed under deletion and contraction.

For every two sets $A, B \subseteq E$, their symmetric difference is denoted by
\[
    A \SymDif B:=(A\setminus B)\cup (B\setminus A).
\]
A fundamental property of binary matroids is that the symmetric difference of two different circuits is a disjoint union of circuits. We note that this property holds only for binary matroids and not for all matroids---in fact, the uniform matroid $U_{2,4}$ does not have this property. A consequence of this property of binary matroids is 
that every binary matroid of nullity at most $d$ contains at most $2^d-1$ circuits, all of which can be enumerated efficiently.

\begin{restatable}{lemma}{binarymatroidnullity}\label{lemma:binary-matroid-nullity}
    Let $M=(E,\mathcal{I})$ be a binary matroid with $m=|E|$ elements and $\nu_M(E)\leq d$. Then, $M$ contains at most $2^d-1$ circuits. Moreover, given independence-oracle access to $M$, all circuits of $M$ can be enumerated using $O(2^d m)$ independence queries and $2^d \poly(m)$ time.
\end{restatable}

We defer the proof of \Cref{lemma:binary-matroid-nullity} to \Cref{appendix:binary-nullity}.

\subsection{Matroid Sums and Seymour's Decomposition Theorem}

Seymour's decomposition theorem constructs every regular matroid from graphic matroids, cographic matroids, and $R_{10}$. We start with the following sum operations.

\begin{definition}[Sum of Matroids~\cite{seymour1980decomposition}]\label{definition:sum-matroid}
    Let $M_1=(E_1, \mathcal{I}_1)$ and $M_2=(E_2, \mathcal{I}_2)$ be two binary matroids and let $S:=E_1\cap E_2$. The \emph{sum} of $M_1$ and $M_2$, denoted by $M_1 \SymDif M_2$, is the binary matroid on $E_1 \SymDif E_2$ whose circuits are the inclusion-wise minimal nonempty sets of the form $C_1 \SymDif C_2$, where $C_i$ is a (possibly empty) disjoint union of circuits of $M_i$ for every $i\in [2]$ and $C_1\cap S=C_2\cap S$.
\end{definition}

The following lemma shows that every set of the form $C_1 \Delta C_2$ as described in \Cref{definition:sum-matroid} is a disjoint union of circuits.

\begin{lemma}\label{lemma:sum-circuit}
    Let $M_1=(E_1, \mathcal{I}_1)$ and $M_2=(E_2, \mathcal{I}_2)$ be two binary matroids and let $S:=E_1\cap E_2$. Let $M:=M_1 \SymDif M_2$ be the sum of $M_1$ and $M_2$. Let $C:= C_1 \SymDif C_2$, where $C_i$ is a (possibly empty) disjoint union of circuits of $M_i$ for every $i\in [2]$. If $C_1\cap S=C_2\cap S$, then $C$ is a disjoint union of circuits of $M$.
\end{lemma}
\begin{proof}
    If $C=\emptyset$, the lemma immediately holds. Otherwise, we choose a minimum-cardinality nonempty subset $D\subseteq C$ admitting a decomposition $D=D_1 \SymDif D_2$, where each $D_i$ is a (possibly empty) disjoint union of circuits of $M_i$ for $i\in [2]$. Such a choice exists because taking $D_1=C_1$ and $D_2=C_2$ gives $C$. By minimality, no proper nonempty subset of $D$ admits such a decomposition. Thus, $D$ is a circuit of $M$ by \Cref{definition:sum-matroid}. Since $D\subseteq C$, we have
    \[
        C\setminus D=C\SymDif D = (C_1\SymDif D_1) \SymDif (C_2\SymDif D_2).
    \]
    We recall that the binary property implies that each $C_i \SymDif D_i$ is again a disjoint union of circuits of $M_i$ for $i\in [2]$. We can therefore repeat the argument on $C\setminus D$. Each repetition removes a nonempty circuit, so the process ultimately terminates and partitions $C$ into disjoint circuits of $M$.
 
\end{proof}

The $1$-, $2$-, and $3$-sums are the following special cases.

\begin{definition}[$1$-, $2$-, and $3$-sums~\cite{gurjar2019number}]\label{definition:matroid-sum}
    Let $M_1=(E_1, \mathcal{I}_1)$ and $M_2=(E_2, \mathcal{I}_2)$ be two binary matroids and let $S:=E_1\cap E_2$. The sum $M_1 \SymDif M_2$ is called
    \begin{enumerate}
        \item a $1$-sum if $|S|=0$,
        \item a $2$-sum if $|S|=1$, $S$ is not a circuit of $M_1$, $M_2$, or their duals, and $|E_1|, |E_2|\geq 3$, and
        \item a $3$-sum if $|S|=3$, $S$ is a circuit of $M_1$ and $M_2$, $S$ does not contain a cocircuit of $M_1$ or $M_2$, and $|E_1|, |E_2|\geq 7$.
    \end{enumerate}
\end{definition}

The following lemma represents circuits of $M_1 \Delta M_2$ in terms of circuits of $M_1$ and $M_2$.

\begin{lemma}[\cite{seymour1980decomposition}]\label{lemma:circuit-sum}
    Let $M_1=(E_1, \mathcal{I}_1)$ and $M_2=(E_2, \mathcal{I}_2)$ be two binary matroids and let $S:=E_1\cap E_2$. Let $\mathcal{C}_1$ and $\mathcal{C}_2$ be the circuits of $M_1$ and $M_2$, respectively. Let $M:=M_1\SymDif M_2$ be a $1$-sum, $2$-sum, or $3$-sum of $M_1$ and $M_2$. Then, for every circuit $C$ of $M$, exactly one of the following holds:
    \begin{enumerate}
        \item $C \in \mathcal{C}_1$ with $S\cap C=\emptyset$;
        \item $C \in \mathcal{C}_2$ with $S\cap C=\emptyset$;
        \item there exist unique $e\in S$, $C_1 \in \mathcal{C}_1$, and $C_2\in \mathcal{C}_2$ such that
        \[
            S\cap C_1 = S\cap C_2 = \{e\} \quad \text{and} \quad C=C_1 \SymDif C_2.
        \]
    \end{enumerate}
\end{lemma}

We need Seymour's decomposition theorem for regular matroids. 
Of particular relevance to this decomposition result is the binary matroid $R_{10}$ consisting of $10$ elements---see Figure \ref{fig:r10} for a binary representation of $R_{10}$. 
\begin{figure}[hbt]
    \centering
    \setlength{\arraycolsep}{3pt}
    \renewcommand{\arraystretch}{1.05}
    \[
    A_{R_{10}}=
    \begin{pmatrix}
        1&1&0&0&1&1&0&0&0&0\\
        1&1&1&0&0&0&1&0&0&0\\
        0&1&1&1&0&0&0&1&0&0\\
        0&0&1&1&1&0&0&0&1&0\\
        1&0&0&1&1&0&0&0&0&1
    \end{pmatrix}.
    \]
    \caption{A binary representation of the matroid \(R_{10}\).}
    \label{fig:r10}
\end{figure}

We state Seymour's decomposition theorem now. 
\begin{theorem}[Seymour's Decomposition Theorem~\cite{seymour1980decomposition}]\label{thm:seymour-decomposition}
    Every regular matroid can be obtained by means of $1$-sums, $2$-sums, and $3$-sums, starting from matroids which are graphic, cographic, or $R_{10}$.
\end{theorem}

Gurjar and Vishnoi~\cite{gurjar2019number} showed the following criterion for the sum operation on regular matroids to be associative.

\begin{lemma}[\cite{gurjar2019number}]\label{lemma:associative}
    Let $M_1$, $M_3$, and $M_4$ be regular matroids with ground set $E_1$, $E_3$, and $E_4$, respectively. Let $M_2:=M_3 \SymDif M_4$ be a $1$-, $2$-, or $3$-sum of $M_3$ and $M_4$ with $S_1:=E_3\cap E_4$, and let $M:=M_1 \SymDif M_2$ be a $1$-, $2$-, or $3$-sum of $M_1$ and $M_2$ with $S_2:=E_1\cap E_2$. Suppose that $S_2\subseteq E_3$, then
    \[
        M=M_1\SymDif (M_3 \SymDif M_4) = (M_1 \SymDif M_3) \SymDif M_4,
    \]
    where $(M_1 \SymDif M_3)$ is defined via the common set $S_2$ and $(M_1 \SymDif M_3) \SymDif M_4$ is defined via the common set $S_1$.
\end{lemma}

The following property of $1$- and $2$-sums by Aprile and Fiorini~\cite{aprile2022regular} will also be used in the later proof.

\begin{lemma}[\cite{aprile2022regular}]\label{lemma:12sum-regular}
    Let $M=M_1 \SymDif M_2$ be a $1$- or $2$-sum of binary matroids. Then, each of $M_1$ and $M_2$ is isomorphic to a matroid obtained from $M$ by deletions and contractions. In particular, if $M$ is regular, then both $M_1$ and $M_2$ are regular.
\end{lemma}
\section{Clean Unordered Decomposition Tree}\label{section:decomposition}

In this section, we establish the existence of a clean unordered decomposition tree associated with a regular matroid. This fact will be used throughout this work.

\subsection{Definitions and Structural Properties}

We start with the formulation of unordered decomposition trees by Gurjar and Vishnoi~\cite{gurjar2019number}. An unordered decomposition tree represents a matroid by assigning a matroid to each node and a sum operation to each edge. The sums may be performed in an arbitrary order, always producing the same matroid. We formalize it as follows.

\begin{definition}[Unordered Decomposition Trees~\cite{gurjar2019number}]\label{definition:unordered-decomposition-tree}
    Let $M=(E,\mathcal{I})$ be a regular matroid. An \emph{unordered decomposition tree} of $M$ consists of a tree $\mathcal{T}$ and a nonempty regular matroid $M_u=(E_u,\mathcal{I}_u)$ associated with every node $u\in V(\mathcal{T})$, satisfying the following properties:
    \begin{enumerate}
        \item For every edge $e=\{u,v\}\in E(\mathcal{T})$, the \emph{boundary set}
        \[
            S_{e}:=E_u \cap E_v
        \]
        has size $0$, $1$, or $3$, and $M_u\SymDif M_v$ is a $1$-, $2$-, or $3$-sum (\Cref{definition:matroid-sum}), respectively. The ground sets of nonadjacent nodes are disjoint.

        \item For every node $u\in V(\mathcal{T})$, the boundary sets $S_e$ over all edges $e$ adjacent to $u$ are pairwise disjoint.

        \item For every nonempty connected subtree $\mathcal{Q}\subseteq \mathcal{T}$, its node matroids can be combined by sums along its edges in an arbitrary order. The resulting matroid, denoted by $M_{\mathcal{Q}}$, is independent of this order and has ground set
        \[
            E_{\mathcal{Q}}:=\SymDif_{v\in V(\mathcal{Q})} E_v.
        \]

        \item The matroid $M_{\mathcal{T}}$ is isomorphic to $M$. In particular, $E=\SymDif_{u\in V(\mathcal{T})} E_u$.
    \end{enumerate}
\end{definition}

We emphasize that the boundary elements belong to the ground sets of matroids associated to both endpoints of a tree edge and disappear when the corresponding sum is performed. Every other element belongs to exactly one node and survives in the assembled matroid. After identifying $M_{\mathcal{T}}$ with $M$, these surviving elements are precisely the elements of $E$. See \Cref{figure:unordered-decomposition-tree} for an example of an unordered decomposition tree. Gurjar and Vishnoi~\cite{gurjar2019number} proved that every regular matroid has a corresponding unordered decomposition tree as follows.

\begin{theorem}[\cite{gurjar2019number}]\label{thm:unordered-decomposition-tree}
    Every regular matroid admits an unordered decomposition tree in which each node matroid is graphic, cographic, or isomorphic to $R_{10}$.
\end{theorem}

\begin{figure}[ht]
    \centering
    \includegraphics[width=0.75\linewidth]{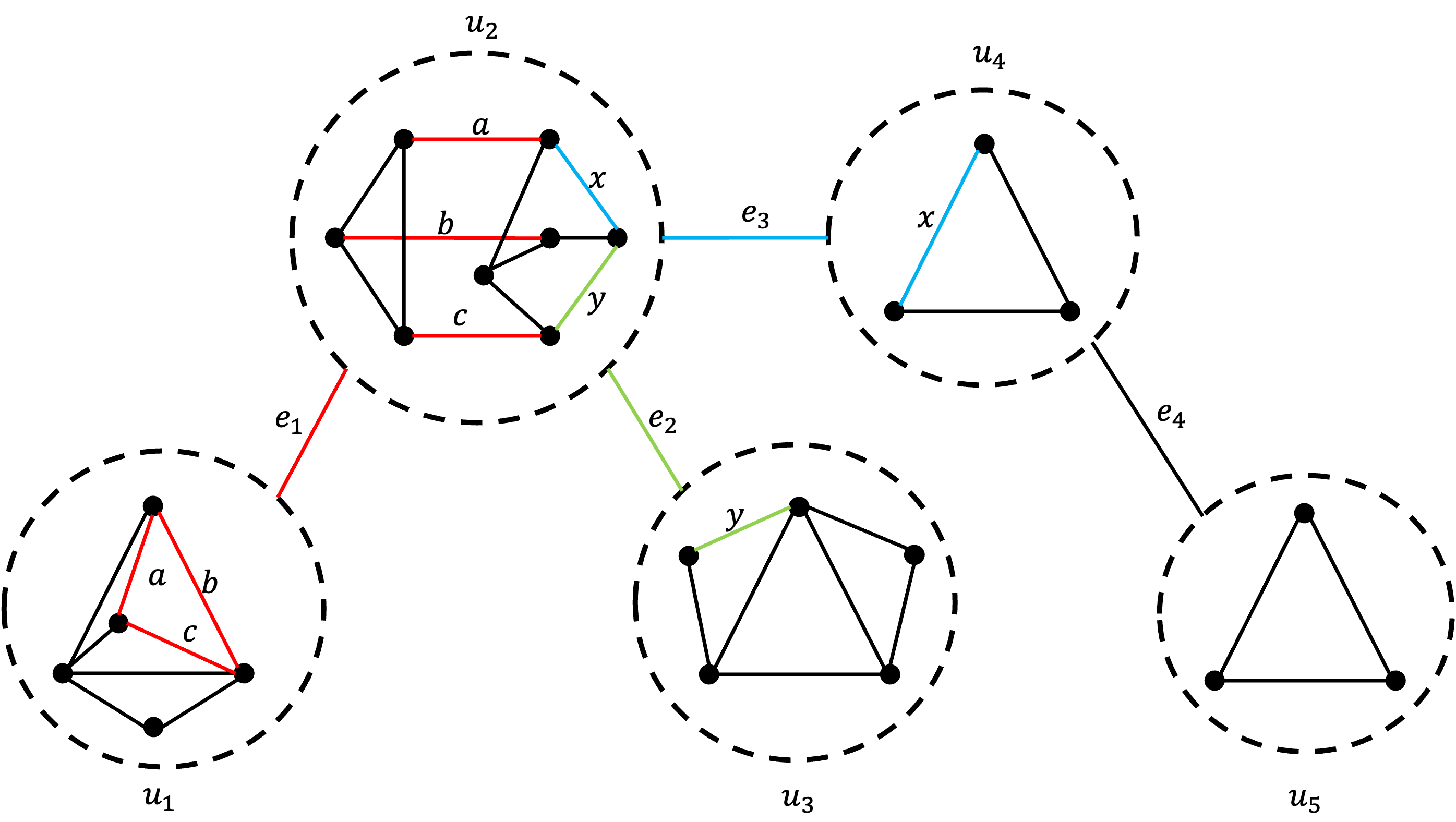}
    \caption{The figure shows an unordered decomposition tree with five nodes $\{u_1, u_2, u_3, u_4, u_5\}$ and four edges $e_1=\{u_1, u_2\}$, $e_2=\{u_2, u_3\}$, $e_3=\{u_2, u_4\}$, and $e_4=\{u_4, u_5\}$. $M_{u_1}$, $M_{u_3}$, $M_{u_4}$, and $M_{u_5}$ are graphic matroids, while $M_{u_2}$ is a cographic matroid that is not graphic. The edge $e_1$ corresponds to a $3$-sum of $M_{u_1}$ and $M_{u_2}$ with $S_{e_1}=\{a,b,c\}$. The edge $e_2$ corresponds to a $2$-sum of $M_{u_2}$ and $M_{u_3}$ with $S_{e_2}=\{y\}$. The edge $e_3$ corresponds to a $2$-sum of $M_{u_2}$ and $M_{u_4}$ with $S_{e_3}=\{x\}$. The edge $e_4$ corresponds to a $1$-sum of $M_{u_4}$ and $M_{u_5}$ with $S_{e_4}=\emptyset$.
    }
    \label{figure:unordered-decomposition-tree}
\end{figure}

We point out that the third property in \Cref{definition:unordered-decomposition-tree} follows from the first two. We call a (possibly empty) disjoint union of circuits a \emph{cycle}. By \Cref{definition:sum-matroid} and \Cref{lemma:sum-circuit}, the cycles of $M_1\SymDif M_2$ are exactly the sets $C_1\SymDif C_2$, where $C_i$ is a cycle of $M_i$ and $C_1\cap S=C_2\cap S$. Since every element belongs to at most two node ground sets, induction on the number of sums shows that performing the sums along the edges of $\mathcal{Q}$ in any order produces the binary matroid whose cycles are exactly the sets $\SymDif_{v\in V(\mathcal{Q})}C_v$, where each $C_v$ is a cycle of $M_v$ and $C_u\cap S_e=C_v\cap S_e$ for every edge $e=\{u,v\}$ of $\mathcal{Q}$. This family does not depend on the order, and a binary matroid is determined by its cycles.

\begin{remark}\label{remark:subtree-sums}
    Let $\mathcal{Q}$ be a connected subtree of $\mathcal{T}$, and let $f=\{x,y\}$ be a tree edge with $x\in V(\mathcal{Q})$ and $y\notin V(\mathcal{Q})$. Then $S_f\subseteq E_{\mathcal{Q}}$, and $S_f$ satisfies in $M_{\mathcal{Q}}$ the conditions of \Cref{definition:matroid-sum} that it satisfies in $M_x$: we have $|E_{\mathcal{Q}}|\geq |E_x|$; if $S_f=\{p\}$, then $p$ is neither a loop nor a coloop of $M_{\mathcal{Q}}$; and if $|S_f|=3$, then $S_f$ is a circuit of $M_{\mathcal{Q}}$ that contains no cocircuit of $M_{\mathcal{Q}}$. Consequently, removing an edge of a connected subtree $\mathcal{Q}$ writes $M_{\mathcal{Q}}$ as a $1$-, $2$-, or $3$-sum of the matroids of the two resulting subtrees. This follows by a routine induction on $|V(\mathcal{Q})|$, splitting $\mathcal{Q}$ along one of its edges and using \Cref{definition:sum-matroid} and \Cref{lemma:sum-circuit}.
\end{remark}

The definition also allows us to combine the matroids represented by two subtrees joined by an edge. Specifically, let $e=\{u,v\}$ be an edge of a connected subtree $\mathcal{Q}$, and let $\mathcal{Q}_u$ and $\mathcal{Q}_v$ be the two subtrees obtained by removing the edge $e$, which contains $u$ and $v$, respectively. We may first construct $M_{\mathcal{Q}_u}$ and $M_{\mathcal{Q}_v}$, then take their sum along the boundary set $S_e$, obtaining
\[
    M_{\mathcal{Q}} = M_{\mathcal{Q}_u} \SymDif M_{\mathcal{Q}_v}.
\]

We show a linear bound on the size of the decomposition as follows.
\begin{lemma}\label{lemma:unordered-decomposition-tree-size}
    Let $M=(E,\mathcal{I})$ be a regular matroid with $m=|E|$. Let $\mathcal{T}$ be an unordered decomposition tree of $M$ and $E_u$ be the ground set of the matroid associated with node $u\in V(\mathcal{T})$. Then,
    \[
        |V(\mathcal{T})|\leq 2m \qquad \text{and} \qquad \sum_{u\in V(\mathcal{T})} |E_u|=O(m).
    \]
\end{lemma}
\begin{proof}
    We first bound the number of nodes. We observe that each element of $E$ belongs to $E_u$ for exactly one node $u\in V(\mathcal{T})$. Hence, at most $m$ nodes $u$ satisfy $E_u\cap E\neq \emptyset$. We prove that every other node $u$ with $E_u\cap E=\emptyset$ has degree at least three.

    Let $u\in V(\mathcal{T})$ be a node with $E_u\cap E=\emptyset$, i.e., the set $E_u$ consists entirely of the boundary sets associated with its incident edges. Suppose that $u$ has degree at most two. We note that each boundary set has at most three elements, which implies that $|E_u|\leq 6$. This further shows that none of its incident edges represents a $3$-sum, since each matroid in a $3$-sum has at least seven elements. Thus, each boundary set has at most one element and $|E_u|\leq 2$. This again shows that none of its incident edges represents a $2$-sum, since each matroid in a $2$-sum has at least three elements. Hence, every incident edge represents a $1$-sum and all boundary sets are empty, which implies that $|E_u|=0$, contradicting the assumption that $M_u$ is a nonempty matroid.

    In particular, every leaf of $\mathcal{T}$ contains at least an element of $E$, which implies that $\mathcal{T}$ has at most $m$ leaves. Since a tree has at most as many nodes of degree at least three as leaves, there are at most $m$ nodes $u\in V(\mathcal{T})$ with $E_u \cap E = \emptyset$. We recall that there are at most $m$ nodes $u$ satisfy $E_u \cap E \neq \emptyset$. Combining these bounds gives
    \[
        |V(\mathcal{T})|\leq 2m.
    \]

    Finally, summing the sizes of the matroids at all nodes counts each element of the assembled matroid once and each boundary element twice. Therefore,
    \[
        \sum_{u\in V(\mathcal{T})} |E_u| = m + 2\sum_{e\in E(\mathcal{T})}|S_e| \leq m+6|E(\mathcal{T})| = m+6(|V(\mathcal{T})|-1) = O(m).
    \]
\end{proof}

The following lemma describes how a circuit of the matroid decomposes into circuits in the matroids associated with nodes and connected subtrees in the unordered decomposition tree.

\begin{lemma}\label{lemma:circuit-projection}
    Let $\mathcal{T}$ be an unordered decomposition tree of a regular matroid $M=(E,\mathcal{I})$, with a matroid $M_u$ on ground set $E_u$ at each node $u\in V(\mathcal{T})$. Every circuit $C$ of $M$ admits a unique decomposition
    \[
        C = \SymDif_{u\in V(\mathcal{T})} C_u,
    \]
    where each node $u$ is assigned a set $C_u\subseteq E_u$ that is either empty or a circuit of $M_u$, which is called the projection of $C$ onto node $u$, satisfying the following properties:
    \begin{enumerate}
        \item For every edge $e=\{u,v\}\in E(\mathcal{T})$, the intersections $C_u\cap S_e=C_v\cap S_e$ and contain at most one element, i.e.,
        \[
            |C_u\cap S_e|=|C_v\cap S_e|\leq 1.
        \]
        \item The nodes with nonempty projections induce a connected subtree.
        \item For every nonempty connected subtree $\mathcal{Q}$, the set
        \[
            C_{\mathcal{Q}}:= \SymDif_{u\in V(\mathcal{Q})} C_u
        \]
        is either empty or a circuit of $M_{\mathcal{Q}}$. We call $C_{\mathcal{Q}}$ the projection of $C$ onto $\mathcal{Q}$.
    \end{enumerate}
\end{lemma}
\begin{proof}
    We first construct the node projections and establish the first two properties. We split the tree $\mathcal{T}$ along an edge and express $M$ as the sums of the matroids represented by the two resulting subtrees, which is a valid sum by \Cref{remark:subtree-sums}. By \Cref{lemma:circuit-sum}, $C$ either belongs to one of these matroids and avoids the boundary set, or splits into two circuits meeting the boundary set in the same single element. We apply this procedure recursively, assigning empty projections throughout any subtree receiving an empty set. Whenever both sides receive nonempty circuits, their shared boundary element occurs at the endpoint of the splitting edge on each side. Thus, the construction satisfies the boundary condition and keeps the nodes with nonempty projections connected.

    We now show that the decomposition is unique. Suppose that sets $D_u$ give another such decomposition. Let $u$ be a leaf of $\mathcal{T}$, and let $S$ be the boundary set on its incident edge. Since every element of $E_u \setminus S$ occurs only at $u$, the two decompositions agree on these elements. Hence,
    \[
        C_u \SymDif D_u \subseteq S \qquad \text{and} \qquad |C_u \SymDif D_u|\leq |C_u \cap S|+|D_u \cap S|\leq 2.
    \]
    We recall that $M_u$ is a binary matroid, which implies that $C_u \SymDif D_u$ is a disjoint union of circuits of $M_u$. However, a singleton boundary is independent, and every proper subset of a three-element boundary is also independent. Therefore, $C_u=D_u$ and the two decompositions have equal symmetric differences over the remaining nodes. We may therefore remove the leaf $u$ and repeat the argument at a leaf of the remaining tree. This proves equality at each successive leaf. At the final node, the two remaining symmetric differences are precisely its two projections and these must also agree. This shows that $C_u=D_u$ for every node $u\in V(\mathcal{T})$.

    Finally, we prove that the third property holds. We fix a connected subtree $\mathcal{Q}$. In the recursive construction above, we may first split along the edges joining $\mathcal{Q}$ to the rest of the tree. These splits isolate $\mathcal{Q}$ without splitting any of its internal edges. The construction therefore assigns to $\mathcal{Q}$ a set $D$ that is either empty or a circuit of $M_{\mathcal{Q}}$. Completing the construction yields the same node projections $C_u$ by uniqueness. Hence,
    \[
        D = \SymDif_{u\in V(\mathcal{Q})}C_u=C_{\mathcal{Q}},
    \]
    proving the third property.
\end{proof}

\subsection{Clean Refinement}

We observe that every boundary set used in a $3$-sum is a three-element circuit in both node matroids. In a graphic matroid, these elements correspond to a triangle; in a cographic matroid, they correspond to a minimal cut. We emphasize that since $R_{10}$ has no three-element circuit, an $R_{10}$ node cannot be incident to a $3$-sum edge. We next strengthen the decomposition by imposing an additional condition at cographic nodes: every boundary set of size three must consist exactly of the three edges incident to a single vertex in the underlying graph.

\begin{definition}[Clean Unordered Decomposition Tree]
    Let $M=(E,\mathcal{I})$ be a regular matroid. An unordered decomposition tree $\mathcal{T}$ of $M$ is \emph{clean} if
    \begin{enumerate}
        \item every node matroid is graphic, cographic, or isomorphic to $R_{10}$, and
        \item for every cographic node $v\in V(\mathcal{T})$ with $M_v=M^*(G_v)$ that is not graphic and every tree edge $e\in E(\mathcal{T})$ incident to $v$ with $|S_e|=3$, there exists a vertex $w\in V(G_v)$ of degree three that is incident to all edges in $S_e$.
    \end{enumerate}

\end{definition}

We remark that the unordered decomposition tree in \Cref{figure:unordered-decomposition-tree} is not clean, since the three elements in $S_{e_1}$ are not incident to a single vertex in the underlying graph of the cographic matroid $M_{u_2}$. The clean property follows from the following theorem once the $1$- and $2$-sum decompositions have been performed.

\begin{theorem}[{\cite[Theorem 6 and Proposition 27]{aprile2022regular}}]\label{theorem:3-connected-decomposition}
    Let $M=(E,\mathcal{I})$ be a regular matroid that is not isomorphic to $R_{10}$ with at least four elements. Suppose that $M$ cannot be expressed as a $1$- or $2$-sum of smaller matroids. Then, $M$ admits a clean unordered decomposition tree in which every edge represents a $3$-sum.
\end{theorem}

We now prove that every regular matroid admits a clean unordered decomposition tree. We first construct a decomposition tree by repeatedly splitting node matroids through $1$- or $2$-sums until neither decomposition is possible. Nodes whose matroids are already graphic, cographic, or isomorphic to $R_{10}$ require no further refinement. Every remaining node matroid satisfies the assumption in \Cref{theorem:3-connected-decomposition}, so we can replace the node with a clean unordered decomposition tree. Reconnecting these trees along the original $1$- and $2$-sum edges yields the desired clean decomposition. We formalize this process in the following proof.

\begin{lemma}\label{lemma:clean-unordered-decomposition-tree}
    Let $M$ be a regular matroid. Then, $M$ admits a clean unordered decomposition tree.
\end{lemma}
\begin{proof}
    We first construct a tree using only $1$- and $2$-sums. We start with a single node whose matroid is $M$. We first perform $1$-sum splits until every node matroid is connected, and only then perform $2$-sum splits. Whenever a node matroid $M_v$ can be written as a $1$- or $2$-sum $M_v=M_{v_1} \SymDif M_{v_2}$ of two smaller nonempty matroids $M_{v_1}$ and $M_{v_2}$, we replace $v$ with two nodes $v_1$ and $v_2$ carrying $M_{v_1}$ and $M_{v_2}$ and join them by a new edge $e=\{v_1,v_2\}$. The boundary set of this edge is
    \[
        S_e:=E_{v_1}\cap E_{v_2}.
    \]
    For a $1$-sum, $S_e=\emptyset$. For a $2$-sum, $S_e$ consists of the single element shared by $M_{v_1}$ and $M_{v_2}$. In this case, we give this new boundary element a label not used elsewhere in the tree.

    We note that each edge previously incident to $v$ must also be reattached. Consider such an edge $f=\{u,v\}$. If $S_f=\emptyset$, we replace $f$ with an edge from $u$ to either $v_1$ or $v_2$, retaining the empty boundary set. If $S_f=\{q\}$, then $q\in E_{v} = E_{v_1}\SymDif E_{v_2}$, which implies that $q$ belongs to exactly one of the two node ground sets. We reattach $f$ to that node, retaining its boundary set $\{q\}$. This preserves pairwise disjointness of the boundary sets at every node.

    According to \Cref{lemma:12sum-regular}, the new node matroids remain regular. We repeat the splitting operation until no node matroid admits a $1$- or $2$-sum decomposition. Each split produces matroids whose overlap with the ground set is smaller, so the process terminates. We denote the resulting tree by $\mathcal{T}_0$. Since a $2$-sum is connected if and only if both of its parts are connected, the $2$-sum splits produce only connected node matroids. Thus, every node incident to a $2$-sum edge is connected and has at least three elements, so it has no loops or coloops. Hence, every edge of $\mathcal{T}_0$ is a valid $1$- or $2$-sum of its two node matroids.

    We next replace each node of $\mathcal{T}_0$ with a clean decomposition of its matroid. For every node $v$, we define a tree $\mathcal{T}_v$ as follows. If $M_v$ is graphic, cographic, or isomorphic to $R_{10}$, let $\mathcal{T}_v$ consist of the single node $v$. Otherwise, $M_v$ is regular, has at least four elements, and admits neither a $1$- nor $2$-sum decomposition. We apply \Cref{theorem:3-connected-decomposition} to obtain a clean decomposition tree $\mathcal{T}_v$ whose sums produce $M_v$.

    For each tree $\mathcal{T}_v$, we preserve the labels of the elements in $E_v$ and give the newly introduced boundary elements fresh labels. Every element of $E_v$ survives the sums within $\mathcal{T}_v$ and therefore belongs to exactly one node of that tree. In particular, this holds for every old boundary element shared with a neighbor of $\{v\}$ in $\mathcal{T}_0$.

    We now connect these trees according to the edges of $\mathcal{T}_0$. For every edge $e=\{u,v\}$, we add an edge with one point in $\mathcal{T}_u$ and the other in $\mathcal{T}_v$, retaining the boundary set $S_e$. If $S_e=\emptyset$, we choose either endpoint arbitrarily within its tree. If $S_e=\{p\}$, we choose the unique node containing $p$ in each tree. We denote the resulting tree by $\mathcal{T}$. If $\mathcal{T}_v$ has more than one node, every node $x$ of $\mathcal{T}_v$ is part of a $3$-sum and is connected (see the proof of \cite[Theorem 6]{aprile2022regular}), so $|E_x|\geq 7$ and $M_x$ has no loops or coloops. Hence, every reattached $2$-sum edge is a valid $2$-sum of its two node matroids. See \Cref{figure:clean-refinement} for an example.

\begin{figure}[ht]
    \centering
    \includegraphics[width=0.7\linewidth]{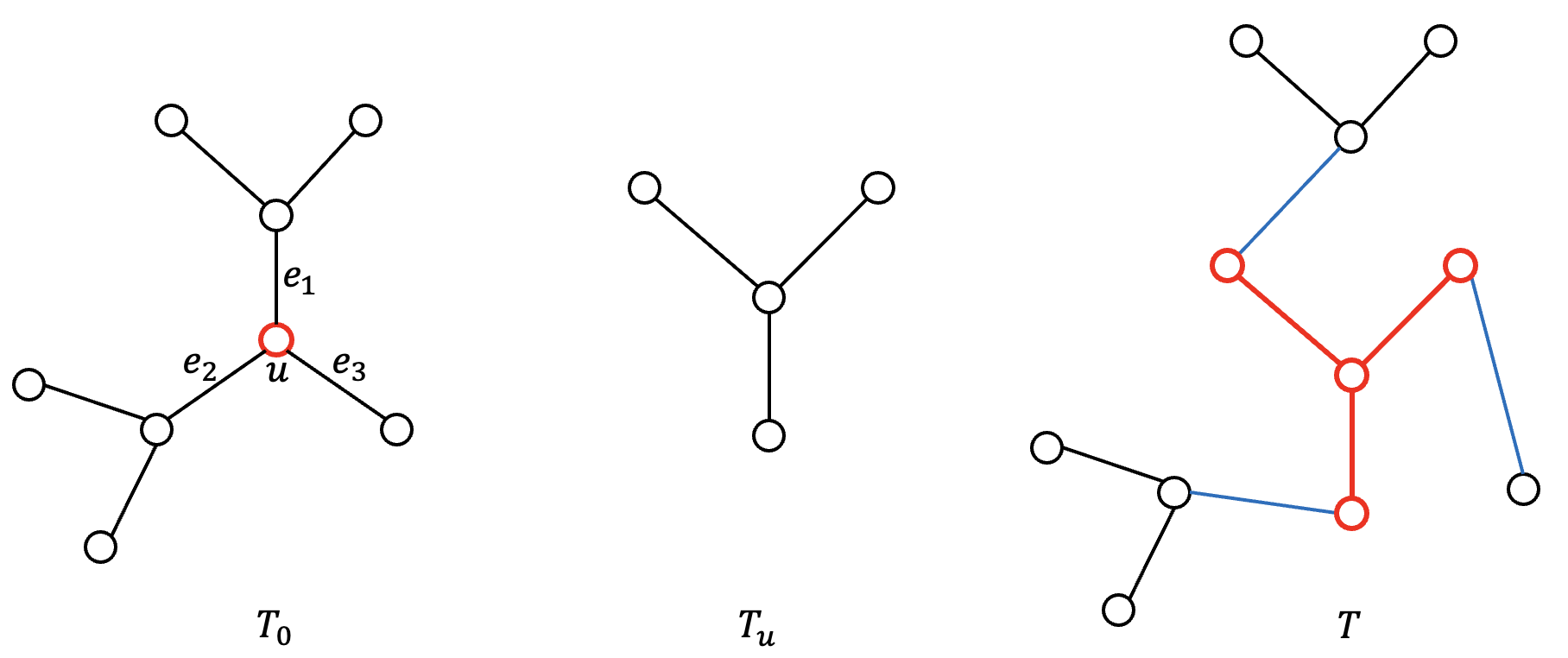}
    \caption{The left figure shows the unordered decomposition tree $\mathcal{T}_0$ of a regular matroid, where each node is associated with a matroid that does not admit a $1$- or $2$-sum decomposition. Moreover, $u$ is the only node in the tree such that matroid $M_u$ is not graphic, cographic, or isomorphic to $R_{10}$. The middle figure illustrates the clean unordered decomposition $\mathcal{T}_u$ of $M_u$. The right figure shows how to replace the node $u$ with $\mathcal{T}_u$, where the original three edges $e_1$, $e_2$, and $e_3$ that are incident to $u$ are now replaced by the three blue edges.
    }
    \label{figure:clean-refinement}
\end{figure}

    We now verify that the boundary sets at each node of $\mathcal{T}$ are pairwise disjoint. We fix a node $x \in V(\mathcal{T}_v)$. An edge incident to $x$ either lies inside $\mathcal{T}_v$ or comes from an edge incident to $v$ in $\mathcal{T}_0$. The boundary sets of the first type are pairwise disjoint because $\mathcal{T}_v$ is an unordered decomposition tree. Those of the second type are pairwise disjoint by the construction of $\mathcal{T}_0$. Moreover, these two collections are disjoint: every boundary element inherited from $\mathcal{T}_0$ belongs to $E_v$, whereas every boundary element of an edge inside $\mathcal{T}_v$ disappears when the sums within $\mathcal{T}_v$ are performed and therefore does not belong to $E_v$.
    
    We observe that performing the sums within each $\mathcal{T}_v$ first recovers $M_v$, after which the sums prescribed by $\mathcal{T}_0$ recover $M$. According to \Cref{lemma:associative}, all these sums may be performed in an arbitrary order. Since each $\mathcal{T}_v$ is clean and the edges between these trees represent only $1$- and $2$-sums, the resulting tree $\mathcal{T}$ is clean.
\end{proof}

\begin{remark}
    We emphasize that a clean unordered decomposition tree can be constructed in polynomial time given access to a rank oracle or an independence oracle for the regular matroid. More precisely, the decomposition trees in \Cref{thm:unordered-decomposition-tree,theorem:3-connected-decomposition} admit polynomial-time constructions using polynomially many oracle queries; see \cite{dinitz2014matroid,gurjar2019number,berczi2024reconfiguration}. We point out that our subsequent proofs use only the existence and structural properties of these trees. The efficient construction is not relevant.
\end{remark}
\section{Matroid Absorption}\label{section:absorption}

In this section, we introduce the \emph{matroid absorption} operation. 
Consider a clean unordered decomposition tree associated with a regular matroid. For a node of the tree, we will obtain a matroid by replacing the components attached to the node with weighted elements. 
The replacement allows us to analyze circuit projections in a modified weighted matroid, with weights accounting for the elements used inside the attached components.

\subsection{Completion and Expansion}\label{section:completion-expansion}

Let $\mathcal{T}$ be a clean unordered decomposition tree of a regular matroid $M=(E, \mathcal{I})$ and $u\in V(\mathcal{T})$ be a node in the tree. Let $M_u=(E_u, \mathcal{I}_u)$ be the matroid associated with node $u$. We observe that removing $u$ and its incident edges leaves several connected components, each joined to $u$ by a unique edge in $\mathcal{T}$. We refer to these as the components attached to $u$.

Let $\mathcal{Q}_1, \ldots, \mathcal{Q}_k$ be components attached to $u$ (not necessarily all of them) such that the unique edge joining each $\mathcal{Q}_i$ to $u$ represents a $2$- or $3$-sum. We select these components for absorption and denote their corresponding boundary sets by $S_1, \ldots, S_k$. All other components are left unabsorbed. Let $\mathcal{Q}$ be the connected subtree formed by $u$, the selected components $\mathcal{Q}_1, \ldots, \mathcal{Q}_k$, and their connecting edges. We work directly with the matroid $M_{\mathcal{Q}}$ throughout this section. We observe that
\[
    E_u \cap E_{\mathcal{Q}} = E_u \setminus \left(\bigcup_{i=1}^{k}S_i\right).
\]
In particular, the boundary elements connecting $u$ to unselected components remain in $M_{\mathcal{Q}}$.

For every $i\in [k]$ and $s\in S_i$, a \emph{completion} of $s$ in $\mathcal{Q}_i$ is a set $P\subseteq E_{\mathcal{Q}_i}\setminus S_i$ such that $P\cup\{s\}$ is a circuit of $M_{\mathcal{Q}_i}$. We note that a completion is always nonempty. The following lemma shows that such a completion always exists.

\begin{lemma}\label{lemma:exist-completion}
    Let $\mathcal{T}$ be a clean unordered decomposition tree of a regular matroid $M$, and let $u\in V(\mathcal{T})$ be a node of the tree. Let $\mathcal{Q}$ be a component attached to $u$ through $2$ or $3$-sum, with boundary set $S$. Then, for every $s\in S$, there exists a completion of $s$ in $\mathcal{Q}$.
\end{lemma}
\begin{proof}
    We recall that $M_{\mathcal{Q}}$ with ground set $E_{\mathcal{Q}}$ is the matroid obtained by performing all sums within $\mathcal{Q}$ in an arbitrary order. The edge joining $\mathcal{Q}$ to $u$ then defines a $2$- or $3$-sum with boundary set $S$. According to \Cref{remark:subtree-sums}, $S$ does not contain a cocircuit of $M_{\mathcal{Q}}$. We recall that an inclusion-wise minimal subset of $S$ that intersects every basis would be a cocircuit. This directly implies that $E_{\mathcal{Q}}\setminus S$ contains a basis of $M_{\mathcal{Q}}$, since otherwise every basis would intersect $S$, contradicting the fact that $S$ does not contain a cocircuit.

    Let $B\subseteq E_{\mathcal{Q}}\setminus S$ be a basis of $M_{\mathcal{Q}}$. We fix $s\in S$. We observe that $B\cup\{s\}$ is dependent and therefore contains a circuit $C$ of $M_{\mathcal{Q}}$. Moreover, $s\in C$ since $B$ is independent. Consequently, $C\setminus \{s\} \subseteq E_{\mathcal{Q}} \setminus S$ is a completion of $s$ in $\mathcal{Q}$.
\end{proof}

Let $P_{\mathcal{Q}_i,s}$ be a completion of $s$ in $\mathcal{Q}_i$ for every $i\in [k]$ and $s\in S_i$. Let $D$ be a circuit of $M_u$ with $|D\cap S_i|\leq 1$ for every $i\in [k]$. The \emph{expansion} of $D$ in $\mathcal{Q}$ is obtained by replacing each element in $s\in D\cap S_i$ with its chosen completion $P_{\mathcal{Q}_i,s}$, i.e.,
\[
    D^+:= \left(D\setminus \left(\bigcup_{i=1}^{k} S_i\right) \right) \cup \left(\bigcup_{i\in [k]}\left( \bigcup_{s\in D\cap S_i}P_{\mathcal{Q}_i,s} \right)\right).
\]
See \Cref{figure:expansion} for an example.

\begin{figure}[ht]
    \centering
    \includegraphics[width=0.85\linewidth]{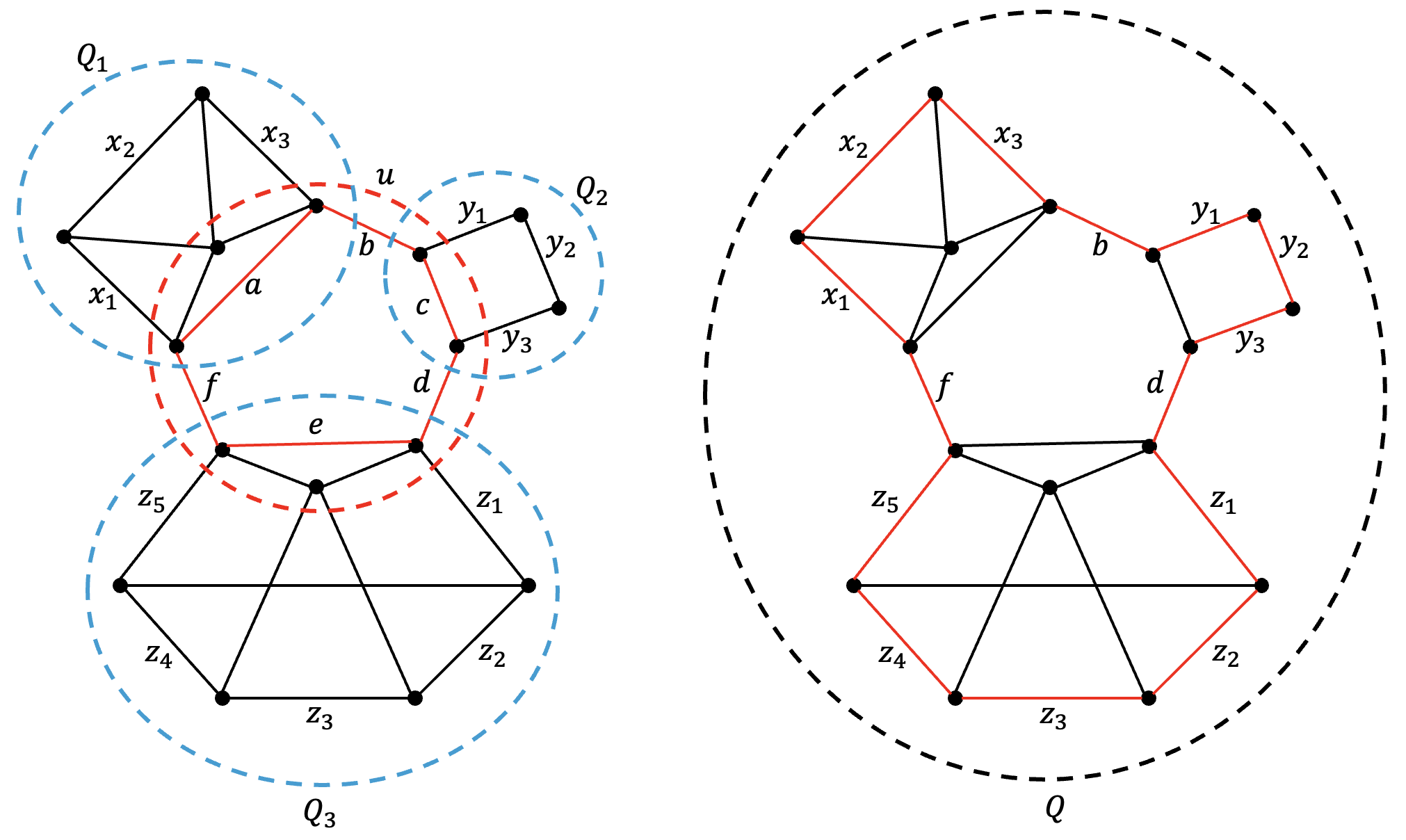}
    \caption{The left picture shows a node $u\in V(\mathcal{T})$ together with three components $\mathcal{Q}_1$, $\mathcal{Q}_2$, and $\mathcal{Q}_3$ attached to $u$. All four matroids $M_u$, $M_{\mathcal{Q}_1}$, $M_{\mathcal{Q}_2}$, and $M_{\mathcal{Q}_3}$ are graphic matroids. Let $D=\{a,b,c,d,e,f\}$ be a circuit of $M_u$. Let $P_{\mathcal{Q}_1,a}:=\{x_1, x_2, x_3\}$ be the completion of $a$ in $\mathcal{Q}_1$, $P_{\mathcal{Q}_2,c}:=\{y_1, y_2, y_3\}$ be the completion of $c$ in $\mathcal{Q}_2$, and $P_{\mathcal{Q}_3,e}:=\{z_1, z_2, z_3, z_4, z_5\}$ be the completion of $e$ in $\mathcal{Q}_3$, respectively. The expansion of $D$ in $\mathcal{Q}$ is shown in the right picture. For readability, the right picture still shows the boundary elements of the three sums (the edges $a$, $c$, $e$ and the remaining triangle edges, in black); they are not elements of $M_{\mathcal{Q}}$.
    }
    \label{figure:expansion}
\end{figure}

We observe that replacing $s\in D\cap S_i$ with its completion $P_{\mathcal{Q}_i,s}$ is equivalent to taking the symmetric difference with the circuit $P_{\mathcal{Q}_i,s}\cup \{s\}$ of $M_{\mathcal{Q}_i}$. Thus, applying \Cref{lemma:sum-circuit} successively to these sums shows that $D^+$ is a disjoint union of circuits of $M_{\mathcal{Q}}$. The following lemma shows that sizes of the minimum-cardinality completions of three elements in a $3$-sum boundary set satisfy the triangle inequalities.

\begin{lemma}\label{lemma:absorb-triangle-inequality}
     Let $\mathcal{T}$ be a clean unordered decomposition tree of a regular matroid $M$, and let $u\in V(\mathcal{T})$ be a node of the tree. Let $\mathcal{Q}$ be a component attached to $u$ through $3$-sum, with boundary set $S$ containing three elements $a$, $b$, and $c$. Let $P_{\mathcal{Q},a}$, $P_{\mathcal{Q},b}$, and $P_{\mathcal{Q},c}$ be the minimum-cardinality completions of $a$, $b$, and $c$ in $\mathcal{Q}$, respectively. Then,
     \[
        |P_{\mathcal{Q},a}|+|P_{\mathcal{Q},b}|\geq |P_{\mathcal{Q},c}|.
     \]
\end{lemma}
\begin{proof}
    According to the definition of completions, $P_{\mathcal{Q},a}\cup \{a\}$, $P_{\mathcal{Q},b}\cup \{b\}$, and $P_{\mathcal{Q},c}\cup \{c\}$ are circuits of matroid $M_{\mathcal{Q}}$. Since $\mathcal{Q}$ is attached to $u$ through $3$-sum, by \Cref{remark:subtree-sums}, $S=\{a,b,c\}$ is also a circuit of matroid $M_{\mathcal{Q}}$. Therefore,
    \[
        \left(P_{\mathcal{Q},a}\cup \{a\}\right) \SymDif \left(P_{\mathcal{Q},b}\cup \{b\}\right) \SymDif \{a,b,c\} = \left( P_{\mathcal{Q},a} \SymDif P_{\mathcal{Q},b} \right) \cup \{c\}
    \]
    is a union of circuits of matroid $M_{\mathcal{Q}}$, which further implies that there exists a set $P\subseteq  P_{\mathcal{Q},a} \SymDif P_{\mathcal{Q},b}$ such that $P\cup \{c\}$ is a circuit of matroid $M_{\mathcal{Q}}$. Therefore,
    \[
        |P_{\mathcal{Q},a}|+|P_{\mathcal{Q},b}|\geq |P|\geq |P_{\mathcal{Q},c}|.
    \]
\end{proof}

We now formalize the absorption lemma as follows.

\begin{lemma}\label{lemma:matroid-absorption}
    Let $\mathcal{T}$ be a clean unordered decomposition tree of a regular matroid $M$, and let $u\in V(\mathcal{T})$ be a node of the tree. Let $\mathcal{Q}_1, \ldots, \mathcal{Q}_k$ be distinct components attached to $u$ through $2$- or $3$-sums, with boundary sets $S_1, \ldots, S_k$. Let $\mathcal{Q}$ be the subtree formed by $u$, these $k$ components, and their connecting edges. For every $i\in [k]$ and $s\in S_i$, let $P_{\mathcal{Q}_i,s}$ be a completion of $s$ in $\mathcal{Q}_i$.
    Suppose that for every $i \in [k]$ with $|S_i|=3$, the completion sizes satisfy the triangle inequalities: 
    \[
        |P_{\mathcal{Q}_i,s}| \leq \sum_{t\in S_i\setminus \{s\}} |P_{\mathcal{Q}_i,t}| \qquad \text{for every } s\in S_i.
    \]
    
    Then, there exists a matroid $\widehat{M_u}=(\widehat{E_u}, \widehat{\mathcal{I}_u})$ with nonnegative weights on elements $w:\widehat{E}_u\rightarrow \R_{\geq 0}$ satisfying the following properties:
    \begin{enumerate}
        \item $E_u\cap E_{\mathcal{Q}} \subseteq \widehat{E_u}$ and $|\widehat{E_u}|= |E_u|$;
        \item There is a bijection between the circuits $D$ of $M_u$ satisfying $|D \cap S_i|\leq 1$ for every $i\in [k]$ and the circuits $\widehat{D}$ of $\widehat{M_u}$. Corresponding circuits $D$ and $\widehat{D}$ satisfy
        \[
            D \cap \left(E_u \cap E_{\mathcal{Q}}\right) = \widehat{D} \cap \left(E_u \cap E_{\mathcal{Q}}\right) \qquad \text{and} \qquad w(\widehat{D}) = |D^+|,
        \]
        where $D^+$ is the expansion of $D$ in $\mathcal{Q}$.
        \item $\widehat{M}_u$ is graphic if $M_u$ is graphic, cographic if $M_u$ is cographic but not graphic, and isomorphic to $R_{10}$ if $M_u$ is isomorphic to $R_{10}$, respectively.
    \end{enumerate}
\end{lemma}

We describe the absorption constructions for graphic and cographic matroids in \Cref{subsection:graphic-absorb,subsection:cographic-absorb}, respectively, and complete the proof of \Cref{lemma:matroid-absorption} in \Cref{subsection:absorb}.

\subsection{Graphic Matroids}\label{subsection:graphic-absorb}

In this section, we prove the correctness of \Cref{lemma:matroid-absorption} when $M_u$ is a graphic matroid.

\begin{lemma}\label{lemma:matroid-absorption-graphic}
    \Cref{lemma:matroid-absorption} holds when $M_u$ is a graphic matroid.
\end{lemma}
\begin{proof}
    Let $G_u=(V_u,E_u)$ be the underlying graph of the graphic matroid $M_u$. We construct a weighted graph $\widehat{G_u}$ by retaining every edge in $E_u\cap E_{\mathcal{Q}}$ with weight $1$ and modifying the selected boundary sets as follows:
    \begin{enumerate}
        \item For every selected component $\mathcal{Q}_i$ attached to $u$ through a $2$-sum with boundary set $S_i=\{s\}$, we retain the edge $s$ and assign it weight
        \[
            w(s):=|P_{\mathcal{Q}_i,s}|.
        \]
        \item For every selected component $\mathcal{Q}_i$ attached to $u$ through a $3$-sum with boundary set $S_i=\{a,b,c\}$, we observe that $a$, $b$, and $c$ form a triangle in $G_u$. Let $v_1$, $v_2$, and $v_3$ be the three vertices of this triangle, where $a=\{v_1,v_2\}$, $b=\{v_2,v_3\}$, and $c=\{v_3,v_1\}$. We remove these three edges in $S_i$, introduce a new vertex $t$, and add the following three new edges
        \[
            x=\{v_1,t\}, \qquad y=\{v_2,t\}, \qquad \text{and} \qquad z=\{v_3,t\}.
        \]
        This $\Delta$ - $Y$ transformation replaces a triangle with a three-edge star. We assign the new edges weight
        $$\begin{aligned}
            w(x) &:= \frac{|P_{\mathcal{Q}_i,a}|+|P_{\mathcal{Q}_i,c}|-|P_{\mathcal{Q}_i,b}|}{2}, \\
            w(y) &:= \frac{|P_{\mathcal{Q}_i,a}|+|P_{\mathcal{Q}_i,b}|-|P_{\mathcal{Q}_i,c}|}{2}, \qquad \text{and} \\
            w(z) &:= \frac{|P_{\mathcal{Q}_i,b}|+|P_{\mathcal{Q}_i,c}|-|P_{\mathcal{Q}_i,a}|}{2},
        \end{aligned}$$
        which are all nonnegative by the assumption.
    \end{enumerate}
    We show an example in \Cref{figure:absorption-graphic}.

\begin{figure}[ht]
    \centering
    \includegraphics[width=0.75\linewidth]{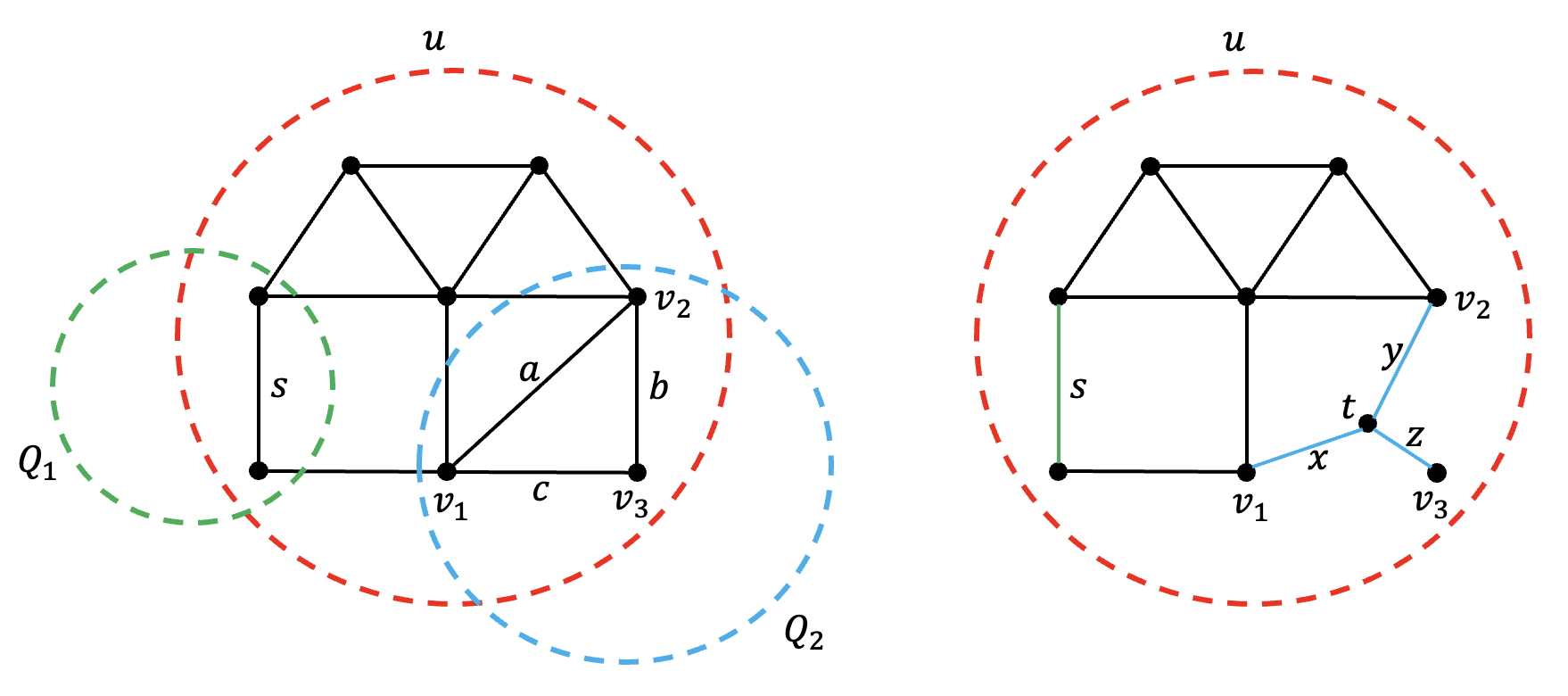}
    \caption{The left picture shows a node $u\in V(\mathcal{T})$ together with two components $\mathcal{Q}_1$ and $\mathcal{Q}_2$ attached to $u$. Here, $\mathcal{Q}_1$ is attached to $u$ via a $2$-sum and $\mathcal{Q}_2$ is attached to $u$ via a $3$-sum. The weighted graph $\widehat{G_u}$ is shown in the right picture, where every black edge has weight $1$.
    }
    \label{figure:absorption-graphic}
\end{figure}

    Let $\widehat{E_u}$ be the edge set of the resulting graph $\widehat{G_u}$, and $\widehat{M_u}:=M(\widehat{G_u})$ be the corresponding graphic matroid. We observe that every edge in $E_u \cap E_{\mathcal{Q}}$ is preserved, elements in boundary sets corresponding to $2$-sums are preserved, and elements in boundary sets corresponding to $3$-sums are replaced with exactly three new elements. Thus,
    \[
        E_u \cap E_{\mathcal{Q}} \subseteq \widehat{E_u}
        \qquad\text{and}\qquad
        |\widehat{E_u}|=|E_u|.
    \]
    This proves the first property in \Cref{lemma:matroid-absorption}, and $\widehat{M_u}$ is graphic by construction.

    We next establish the circuit bijection. Let $D$ be a circuit of $M_u$ satisfying $|D\cap S_i|\leq 1$ for every $i\in [k]$, which corresponds to a cycle in graph $G_u$. For every selected boundary set $S_i$ of size three that meets $D$, we replace the unique edge in $D\cap S_i$ with the two-edge path in the corresponding new star joining the same endpoints. All other edges of $D$ remain unchanged. We denote the resulting edge set by $\widehat{D}$. Since $D$ uses at most one edge of every selected boundary set of size three, each replacement introduces a distinct new vertex along the cycle. Therefore, $\widehat{D}$ is a cycle in $\widehat{G_u}$.

    Conversely, a cycle $\widehat{D}$ in $\widehat{G_u}$ uses either zero or two edges of each new star. We replace each such two-edge path with the edge of the corresponding triangle joining its endpoints. This recovers a cycle $D$ in $G_u$ satisfying $|D\cap S_i|\leq 1$ for every $i\in [k]$. We observe that these operations are inverses and leave all edges in $E_u\cap E_{\mathcal{Q}}$ unchanged, establishing the required bijection.

    Finally, we verify that the weight of $\widehat{D}$ equals the size of the expansion $D^+$. We observe that for every selected boundary set $S_i=\{a,b,c\}$ of size three, the assigned weights satisfy
    \[
        w(x)+w(y) = |P_{\mathcal{Q}_i,a}|, \qquad w(y)+w(z) = |P_{\mathcal{Q}_i,b}|, \qquad \text{and} \qquad w(z)+w(x)=|P_{\mathcal{Q}_i,c}|.
    \]
    Together with the weights assigned to the retained edges, these imply that
    \[
        w(\widehat{D}) = |D\cap E_{\mathcal{Q}}| + \sum_{i=1}^{k} \left( \sum_{s\in D\cap S_i} |P_{\mathcal{Q}_i,s}| \right) = |D^+|,
    \]
    where the last equality holds since the expansion uses at most one completion from each attached component, and these completions are pairwise disjoint and also disjoint from $E_u\cap E_{\mathcal{Q}}$.
\end{proof}

\subsection{Cographic Matroids}\label{subsection:cographic-absorb}

We now prove the correctness of \Cref{lemma:matroid-absorption} when $M_u$ is a cographic matroid.

\begin{lemma}\label{lemma:matroid-absorption-cographic}
    \Cref{lemma:matroid-absorption} holds when $M_u$ is a cographic matroid but not graphic.
\end{lemma}
\begin{proof}
    Let $G_u=(V_u, E_u)$ be the underlying graph of the cographic matroid $M_u$, so that $M_u=M^*(G_u)$. We construct a weighted graph $\widehat{G_u}$ by retaining every edge in $E_u\cap E_{\mathcal{Q}}$ with weight $1$ and modifying the selected boundary sets as follows:
    \begin{enumerate}
        \item For every selected component $\mathcal{Q}_i$ attached to $u$ through a $2$-sum with boundary set $S_i=\{s\}$, we retain the edge $s$ and assign it weight
        \[
            w(s):=|P_{\mathcal{Q}_i,s}|.
        \]
        \item For every selected component $\mathcal{Q}_i$ attached to $u$ through a $3$-sum with boundary set $S_i=\{a,b,c\}$, we observe that $a$, $b$, and $c$ form a cut in $G_u$ and the clean property ensures that $S_i$ consists exactly of the three edges incident to a degree-three vertex $t$. Let $v_1$, $v_2$, and $v_3$ be the other endpoints of edges $a$, $b$, and $c$, respectively. We note that $v_1$, $v_2$, and $v_3$ are distinct. Otherwise, two edges of $S_i$ would form a cycle in $G_u$, and hence a cocircuit of $M_u$ contained in $S_i$, contradicting \Cref{definition:matroid-sum}. We remove $t$ together with its three incident edges, and add the new edges
        \[
            x=\{v_1, v_3\}, \qquad y=\{v_1,v_2\}, \qquad \text{and} \qquad z=\{v_2,v_3\}.
        \]
        This $Y$ - $\Delta$ transformation replaces the three-edge star with a triangle. We assign the new edges weight
        $$\begin{aligned}
            w(x) &:= \frac{|P_{\mathcal{Q}_i,a}|+|P_{\mathcal{Q}_i,c}|-|P_{\mathcal{Q}_i,b}|}{2}, \\
            w(y) &:= \frac{|P_{\mathcal{Q}_i,a}|+|P_{\mathcal{Q}_i,b}|-|P_{\mathcal{Q}_i,c}|}{2}, \qquad \text{and} \\
            w(z) &:= \frac{|P_{\mathcal{Q}_i,b}|+|P_{\mathcal{Q}_i,c}|-|P_{\mathcal{Q}_i,a}|}{2},
        \end{aligned}$$
        which are all nonnegative by the assumption.
    \end{enumerate}
    We show an example in \Cref{figure:absorption-cographic}.

\begin{figure}[ht]
    \centering
    \includegraphics[width=0.75\linewidth]{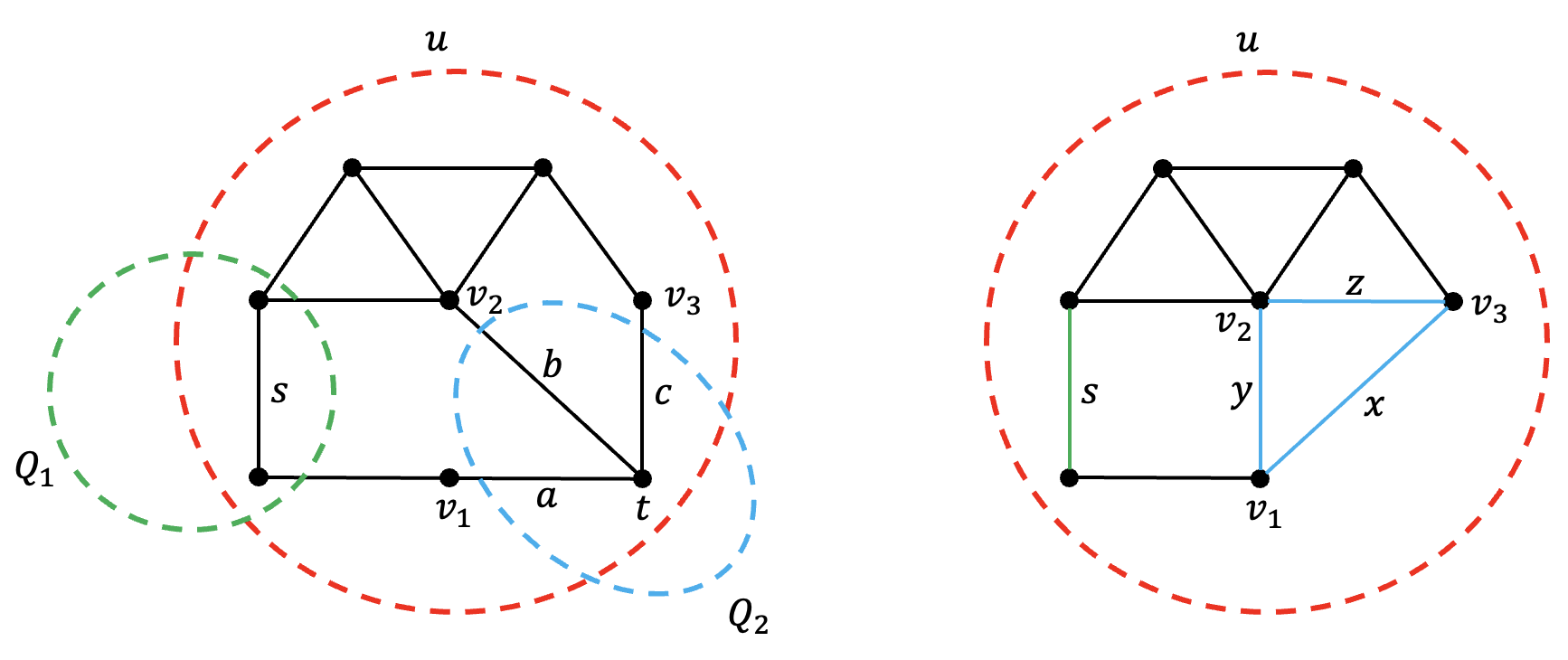}
    \caption{The left picture shows a node $u\in V(\mathcal{T})$ together with two components $\mathcal{Q}_1$ and $\mathcal{Q}_2$ attached to $u$. Here, $\mathcal{Q}_1$ is attached to $u$ via a $2$-sum and $\mathcal{Q}_2$ is attached to $u$ via a $3$-sum. The edges $a, b, c$ are in the ground set of $M_{\mathcal{Q}_2}$. The weighted graph $\widehat{G_u}$ is shown in the right picture, where every black edge has weight $1$.
    }
    \label{figure:absorption-cographic}
\end{figure}

    Let $\widehat{E_u}$ be the edge set of $\widehat{G_u}$, and $\widehat{M_u}:=M^*(\widehat{G_u})$ be the corresponding cographic matroid. We observe that every edge in $E_u\cap E_{\mathcal{Q}}$ is preserved, elements in boundary sets corresponding to $2$-sums are preserved, and elements in boundary sets corresponding to $3$-sums are replaced with exactly three new elements. Thus,
    \[
        E_u\cap E_{\mathcal{Q}}\subseteq \widehat{E_u} \qquad \text{and} \qquad |\widehat{E_u}|=|E_u|.
    \]
    This proves the first property in \Cref{lemma:matroid-absorption}, and $\widehat{M_u}$ is cographic by construction.

    We next establish the circuit bijection. Let $D$ be a circuit of $M_u$ satisfying $|D\cap S_i|\leq 1$ for every $i\in [k]$, which corresponds to a bond (inclusion-wise minimal nonempty cut) of $G_u$. Since $D$ is a cut of $G_u$, we can divide $V_u$ into two sets such that $D$ consists exactly of the edges joining them. We refer to these sets as the two sides of the cut.

    For each selected three-element boundary, the condition $|D\cap S_i|\leq 1$ implies that the removed vertex lies on the same side as at least two of its three neighbors. Removing these vertices leaves two sets of vertices in $\widehat{G_u}$, which defines a cut $\widehat{D}$. If $|D\cap S_i|=1$, exactly one neighbor of the removed vertex lies on the opposite side, and the two edges of the new triangle incident to that neighbor belong to $\widehat{D}$. If $|D\cap S_i|=0$, no edge of the new triangle belongs to $\widehat{D}$. All other edges of $D$ remain unchanged. In particular, $\widehat{D}$ is nonempty.

    We now show that $\widehat{D}$ is a circuit of $\widehat{M_u}$. Suppose that it contains a nonempty proper cut $\widehat{D'}$. We recall that removed vertices are pairwise nonadjacent and all their neighbors remain in $\widehat{G_u}$. Thus, we can restore each removed vertex on the side of $\widehat{D'}$ containing at least two of its neighbors. We denote the resulting cut of $G_u$ by $D'$. Since a cut uses either zero or two edges of a triangle, $\widehat{D'}$ uses either the same pair as $\widehat{D}$ on every new triangle or no edge. Thus, $D'$ is a nonempty proper subset of $D$, contradicting the inclusion-wise minimality of $D$. Therefore, $\widehat{D}$ is a circuit of $\widehat{M_u}$.

    Conversely, let $\widehat{D}$ be a circuit of $\widehat{M_u}$. We partition the vertices of $\widehat{G_u}$ into the two sides of this cut, and restore each removed vertex on the side containing at least two of its three neighbors. This produces a nonempty cut $D$ of $G_u$ satisfying $|D\cap S_i|\leq 1$ for every $i\in [k]$. Suppose that $D$ contains a nonempty proper cut $D'$. Since $D'\subseteq D$, it also meets every selected boundary set in at most one element. Removing the restored vertices from the two sides of $D'$ therefore produces a cut $\widehat{D'}$ in $\widehat{G_u}$. The same edge-replacement rule shows that $\widehat{D'}$ is a nonempty proper subset of $\widehat{D}$, contradicting its inclusion-wise minimality. Hence, $D$ is a circuit of $M_u$. We observe that these operations are inverses and leave all edges in $E_u\cap E_{\mathcal{Q}}$ unchanged, establishing the required bijection.

    Finally, we verify that the weight of $\widehat{D}$ equals the size of the expansion $D^+$. We observe that for every selected boundary set $S_i=\{a,b,c\}$ of size three, the assigned weights satisfy
    \[
        w(x)+w(y) = |P_{\mathcal{Q}_i,a}|, \qquad w(y)+w(z) = |P_{\mathcal{Q}_i,b}|, \qquad \text{and} \qquad w(z)+w(x)=|P_{\mathcal{Q}_i,c}|.
    \]
    Together with the weights assigned to the retained edges, these imply that
    \[
        w(\widehat{D}) = |D\cap E_{\mathcal{Q}}| + \sum_{i=1}^{k} \ \sum_{s\in D\cap S_i} |P_{\mathcal{Q}_i,s}| = |D^+|,
    \]
    where the last equality holds since the expansion uses at most one completion from each attached component, and these completions are pairwise disjoint and also disjoint from $E_u\cap E_{\mathcal{Q}}$.
\end{proof}

\subsection{Proof of \Cref{lemma:matroid-absorption}}\label{subsection:absorb}

We now finish the proof of \Cref{lemma:matroid-absorption} as follows.

\begin{proof}[Proof of \Cref{lemma:matroid-absorption}]
    Based on \Cref{lemma:matroid-absorption-graphic,lemma:matroid-absorption-cographic}, the statement holds when $M_u$ is graphic or cographic. It remains to consider the case where $M_u$ is isomorphic to $R_{10}$. Since $R_{10}$ contains no circuits of size three, every selected component is attached to $u$ through a $2$-sum. Thus, each selected boundary set is a singleton. Let $S_i:=\{s_i\}$ for every $i\in [k]$.

    We set $\widehat{M_u}:=M_u$. Moreover, we assign weight $1$ to every element in $E_u\cap E_{\mathcal{Q}}$ and set
    \[
        w(s_i):=|P_{\mathcal{Q}_i,s_i}| \qquad \text{for every } i\in[k].
    \]
    Since the ground set is unchanged, the first property holds, and $\widehat{M_u}$ remains isomorphic to $R_{10}$.

    We observe that every circuit $D$ of $M_u$ satisfies $|D\cap S_i|\leq 1$ because every $S_i$ is a singleton. Therefore, the correspondence $\widehat{D}=D$ gives the required circuit bijection and preserves all elements in $E_u\cap E_{\mathcal{Q}}$. Moreover,
    \[
        w(\widehat{D}) = |D\cap E_{\mathcal{Q}}| + \sum_{i=1}^{k} \ \sum_{s\in D\cap S_i} |P_{\mathcal{Q}_i,s}| = |D^+|,
    \]
    where the last equality follows because the retained elements and the completions used in $D^+$ are pairwise disjoint.
\end{proof}

\section{Random Deletion returns an Independent Set}\label{section:target-free}
In this section, we analyze the random deletion process to show that it returns an independent set with high probability and thereby prove \Cref{theorem:target-free}. In \Cref{section:target-free-event}, we root a clean unordered decomposition tree and capture good events by defining the notion of $R$-goodness of a node for a set $R\subseteq E$. In \Cref{section:target-free-absorbing}, we use matroid absorption to reduce the analysis at a node with $R$-good children to a weighted matroid that is graphic, cographic, or isomorphic to $R_{10}$. We establish the $R$-goodness using certificates: if no certificate survives, an induction from the leaves to the root shows that every node is $R$-good and hence the remaining set is independent. In \Cref{section:target-free-graphic,section:target-free-cographic}, we construct the certificates and bound their survival probabilities for graphic and cographic nodes, respectively. In \Cref{section:target-free-proof}, we handle $R_{10}$ nodes and combine these bounds to show that no certificate survives with high probability and finish the proof.

We emphasize that we use different certificate constructions for graphic and cographic nodes. At cographic nodes, we take circuit expansions as certificates. These circuits correspond to minimal cuts in the underlying graph, and the fact that the symmetric difference of two cuts is again a cut allows us to apply Karger's cut counting bound (\Cref{theorem:counting-cut}). However, at graphic nodes, the circuits correspond to cycles and the symmetric difference of two cycles need not be a cycle; it may be a disjoint union of cycles, to which Subramanian's cycle counting bound (\Cref{theorem:counting-cycle}) does not directly apply. We therefore use expansions of individual edges and suitably chosen paths as certificates.

\subsection{Rooted Decomposition and Good Events}\label{section:target-free-event}

Let $M=(E,\mathcal{I})$ be a regular matroid with $m\geq 2$ elements and minimum circuit size $g$, and let $R\subseteq E$ be a subset of the ground set. By \Cref{lemma:clean-unordered-decomposition-tree}, let $\mathcal{T}$ be a clean unordered decomposition tree of $M$ and we root it at an arbitrary node $r$. For every node $u\in V(\mathcal{T})$, let $\mathcal{Q}_u$ be the subtree consisting of $u$ and all its descendants. We emphasize the difference between $M_u$ and $M_{\mathcal{Q}_u}$: $M_u$ is the matroid associated with node $u$, whereas $M_{\mathcal{Q}_u}$ is obtained by combining the matroids at $u$ and all its descendants through the sums specified by the edges of $\mathcal{Q}_u$.
Let $S_u$ be the boundary set on the edge joining $u$ to its parent for every node $u\in V(\mathcal{T})\setminus \{r\}$ and $S_r:=\emptyset$. We observe that for every node $u\in V(\mathcal{T})$,
\[
    E_{\mathcal{Q}_u} \setminus S_u = E_{\mathcal{Q}_u}\cap E.
\]
Thus, within $M_{\mathcal{Q}_u}$, the elements outside $S_u$ belong to $E$, whereas the boundary elements in $S_u$ do not belong to $E$ and hence cannot belong to $R$.

We now define $R$-goodness for every node in the decomposition tree, which rules out surviving circuits inside its subtree and controls the completions that can connect the subtree to its parent.

\begin{definition}\label{definition:good-events}
    Let $M=(E,\mathcal{I})$ be a regular matroid with $m\geq 2$ elements and minimum circuit size $g$. Let $R\subseteq E$ be a subset of the ground set, and $\mathcal{T}$ be a clean unordered decomposition tree of $M$. For every node $u\in V(\mathcal{T})$, we say that $\mathcal{G}_u^R$ holds if the following properties hold:
    \begin{enumerate}
        \item No circuit of $M_{\mathcal{Q}_u}$ avoiding $S_u$ is contained in $R$.
        \item For every $s\in S_u$, every completion $P$ of $s$ in $\mathcal{Q}_u$ with $P\subseteq R$ has size $|P|<g/2$.
    \end{enumerate}
    We call the node $u$ to be \emph{$R$-good} if $\mathcal{G}_u^R$ holds.
\end{definition}

We observe that $M_{\mathcal{Q}_r}=M$ and $S_r=\emptyset$, which implies that the root is $R$-good if and only if $R$ is independent. The following lemma shows that a completion of size less than $g/2$, if one exists, is uniquely determined.

\begin{lemma}\label{lemma:unique-completion}
    Let $\mathcal{T}$ be a clean unordered decomposition tree of a regular matroid $M=(E,\mathcal{I})$ with minimum circuit size $g$. For every node $u\in V(\mathcal{T})$ and every $s\in S_u$, there is at most one completion $P$ of $s$ in $\mathcal{Q}_u$ satisfying $|P|<g/2$.
\end{lemma}
\begin{proof}
    We fix a node $u\in V(\mathcal{T})$ and an element $s\in S_u$ in the boundary set. Suppose that there are two distinct completions $P$ and $P'$ of $s$ with $|P|, |P'|<g/2$. Since $M_{\mathcal{Q}_u}$ is binary, the symmetric difference
    \[
        (P\cup \{s\}) \SymDif (P' \cup \{s\}) = P \SymDif P'
    \]
    is a nonempty disjoint union of circuits of $M_{\mathcal{Q}_u}$. We observe that this set avoids $S_u$. According to \Cref{lemma:sum-circuit}, $P \SymDif P'$ is also a disjoint union of circuits of $M$. Consequently,
    \[
        |P\SymDif P'|\leq |P|+|P'|<g,
    \]
    which contradicts the definition of $g$.
\end{proof}

\subsection{Absorbing Child Components}\label{section:target-free-absorbing}

In this section, we use matroid absorption to replace the child components of each node with weighted elements. Then, we show that when all children are $R$-good, the parental $R$-good failure can be certified by a circuit in this weighted matroid.

We fix a node $u\in V(\mathcal{T})$ in the clean unordered decomposition tree. For every child $v$ attached through a $2$- or $3$-sum and every element $s\in S_v$ in the corresponding boundary set, we choose a minimum-cardinality completion $P_{\mathcal{Q}_v,s}$. We emphasize that such completions exist by \Cref{lemma:exist-completion} and that these choices are fixed before the random deletion process. By \Cref{lemma:absorb-triangle-inequality}, the chosen minimum-cardinality completions satisfy the triangle inequalities for every child boundary set of size three. We may therefore apply \Cref{lemma:matroid-absorption} to these child components using the chosen completions and denote the resulting weighted matroid by $\widehat{M_u}$.

The parent boundary $S_u$ is left unabsorbed, and its elements are preserved. We assign weight $0$ to these elements since they are not sampled in \Cref{algo:random-delete} and retain all other weights from the absorption construction. Let $w_u:\widehat{E_u}\rightarrow \R_{\geq 0}$ be the adjusted weights on elements in $\widehat{M_u}$. For every circuit $D$ of $M_u$ meeting each selected child boundary in at most one element, let $\widehat{D}$ be its corresponding circuit in $\widehat{M_u}$, and let $D^+$ be its expansion through the selected child components. According to \Cref{lemma:matroid-absorption},
\[
    \widehat D\cap S_u = D\cap S_u
    \qquad \text{and} \qquad
    w_u(\widehat{D})=|D^+\setminus S_u|.
\]
Thus, the weight counts exactly the elements of $E$ in the expansion. We emphasize that both $\widehat{M_u}$ and $w_u$ depend only on the decomposition and the chosen completions and are unrelated to the random deletion process. The following lemma explains how to use this construction to capture parental $R$-good failure when all children are $R$-good.

\begin{lemma}\label{lemma:target-free-absorp}
    Let $\mathcal{T}$ be a clean unordered decomposition tree of a regular matroid $M=(E,\mathcal{I})$ with minimum circuit size $g$. Let $R\subseteq E$ be a subset of the ground set and $u\in V(\mathcal{T})$ be a node in the tree. Suppose that $u$ is not $R$-good but all its children are $R$-good. Then, $\widehat{M_u}$ contains a circuit $\widehat{D}$ with corresponding circuit $D$ of $M_u$ satisfying the following:
    \begin{enumerate}
        \item $D$ meets every child boundary $S_v$ in at most one element,
        \item $|D\cap S_u|\leq 1$, $D^+\setminus S_u\subseteq R$, and $w_u(\widehat{D})\geq g/2$, and
        \item if $D\cap S_u=\emptyset$, then $w_u(\widehat{D})\geq g$.
    \end{enumerate}
\end{lemma}
\begin{proof}
    Since $u$ is not $R$-good and $\mathcal{G}_u^R$ fails, there exists a circuit $C$ of $M_{\mathcal{Q}_u}$ witnessing the failure. According to \Cref{definition:good-events}, either $C$ avoids $S_u$ and is contained in $R$, or $C$ consists of one element of $S_u$ and a completion contained in $R$ of size at least $g/2$. By \Cref{lemma:circuit-projection}, $C$ has a projection $D$ onto $u$ and a projection onto each child subtree. The projection $D$ is nonempty: otherwise, $C$ would lie entirely in one child subtree and avoid its boundary, contradicting that the child is $R$-good. Moreover, \Cref{lemma:circuit-projection} ensures that $|D\cap S_v|\leq 1$ for every child $v$ of $u$.
    
    We consider a nonempty projection onto a child subtree $\mathcal{Q}_v$. It cannot avoid $S_v$, since otherwise it would be a circuit contained in $R$ in $M_{\mathcal{Q}_v}$, contradicting $\mathcal{G}_v^R$. It therefore meets $S_v$ in a single element $s$, and its remaining elements form a completion $P$ of $s$ contained in $R$. Because $v$ is $R$-good,
    \[
        |P_{\mathcal{Q}_v,s}|\leq |P|<g/2.
    \]
    \Cref{lemma:unique-completion} implies that $P=P_{\mathcal{Q}_v,s}$. Thus, every nonempty child projection uses exactly the selected minimum-cardinality completion. We remark that children attached through $1$-sums have empty projections.

    Consequently, expanding $D$ using the fixed completions recovers $C$, i.e., $D^+=C$. Moreover, $D\cap S_u=C\cap S_u$, which implies that $|D\cap S_u|\leq 1$ and $D^+\setminus S_u\subseteq R$. The weight identity gives
    \[
        w_u(\widehat{D})=|C\setminus S_u|.
    \]
    If $|C\cap S_u|\geq 1$, this quantity is at least $g/2$ by the choice of $C$. If $|C\cap S_u|=0$, applying \Cref{lemma:sum-circuit} with the empty set outside $\mathcal{Q}_u$ shows that $C$ contains a circuit of $M$. Hence, $|C|\geq g$.
\end{proof}

\paragraph{Proof Outline for Theorem \ref{theorem:target-free}.} \Cref{lemma:target-free-absorp} certifies  failure of $R$-goodness at a node with all children being $R$-good through a surviving circuit expansion of sufficiently large weight. To bound the probability of such a failure, we will identify fixed subsets of these expansions whose survival is unlikely. For every node $u\in V(\mathcal{T})$, the following subsections will specify a family $\mathcal{W}_u$ of nonempty subsets of $E_{\mathcal{Q}_u}\setminus S_u$. We call these sets \emph{certificates} at $u$ and denote their collection by $\mathcal{W}_u$. Each certificate consists of elements of $M$ inside the subtree $\mathcal{Q}_u$ and is chosen before random deletion. A certificate $W$ survives if all its elements remain, i.e., $W\subseteq R$. The certificates will be chosen to guarantee the following: if all children of $u$ are $R$-good and no certificate at $u$ survives, then $u$ is $R$-good. Consequently, if no certificate survives anywhere in the tree, induction from the leaves to the root shows that every node is $R$-good and hence $R$ is independent. When $R$ is the remaining set returned by \Cref{algo:random-delete}, every fixed certificate $W$ satisfies
\[
    \Pr[W\subseteq R]=m^{-40|W|/g}.
\]
We will count certificates according to their sizes and apply a union bound to show that none survive with high probability.

\subsection{Certificates at Graphic Nodes}\label{section:target-free-graphic}

In this section, we construct certificates at graphic nodes and bound their survival probabilities. 
Let $u\in V(\mathcal{T})$ be a graphic node, with the graphic matroid $M_u=M(G_u)$ associated with $u$ and the underlying graph $G_u=(V_u,E_u)$. We recall that for every child $v$ attached through a $2$- or $3$-sum and every element $s\in S_v$ in the corresponding boundary set, $P_{\mathcal{Q}_v,s}$ is the completion of $s$ in $\mathcal{Q}_v$ with minimum cardinality. Let $S_u$ be the parent boundary set. As discussed in \Cref{section:target-free-absorbing}, absorbing the child components produces a weighted graphic matroid $\widehat{M_u}=M(\widehat{G_u})$. We emphasize that graphs $G_u$ and $\widehat{G_u}$ may have parallel edges. 

Let $H_u$ be the graph obtained from $G_u$ by deleting the parent boundary edges $S_u$. We observe that every edge of $H_u$ is either an element of $E$ or an element of exactly one child boundary set. For every $e\in E(H_u)$, we define
\[
    A_e:=
    \begin{cases}
        \{e\} & \text{if }e\in E_u\cap E,\\
        P_{\mathcal Q_v,e} & \text{if }e\in S_v
        \text{ for a child }v\text{ of }u,
    \end{cases}
    \qquad \text{and} \qquad
    \ell_u(e):=|A_e|.
\]
We note that $A_e\subseteq E_{\mathcal{Q}_u}\setminus S_u$ and $\ell_u(e)>0$ for every $e\in E(H_u)$.

We call an edge set $F\subseteq E(H_u)$ \emph{admissible} if $|F\cap S_v|\leq 1$ for every child $v$ of $u$. We further define its expansion and length by
\[
    F^+:=\bigcup_{e\in F}A_e
    \qquad\text{and}\qquad
    \ell_u(F):=\sum_{e\in F}\ell_u(e).
\]
We emphasize that $F^+$ uses at most one completion from each child, and distinct child components contain disjoint subsets of $E$, which implies that $|F^+|=\ell_u(F)$. We call a path \emph{simple} if it does not contain repeated vertices. The following lemma shows that an admissible simple path of length less than $g/2$ is uniquely determined by its endpoints. The lemma is trivial if $u$ is the root node of $\mathcal{T}$ and has no children, and it should be seen as a generalization of this special case.

\begin{lemma}\label{lemma:graphic-short-path}
    For every pair of distinct vertices $x,y \in V(H_u)$, there is at most one admissible simple $x$-$y$ path $F$ in $H_u$ with length $\ell_u(F)<g/2$.
\end{lemma}
\begin{proof}
    Suppose that there are two distinct admissible simple $x$-$y$ paths $F_1$ and $F_2$ with length less than $g/2$. We now show that $F_1^+\SymDif F_2^+$ is a nonempty disjoint union of circuits of $M$.

    We observe that every vertex has even degree in $F_1 \SymDif F_2$, which implies that this edge set is a disjoint union of cycles in $G_u$, and hence a disjoint union of circuits of $M_u$. We recall that the expansion of each admissible path can be expressed as its symmetric difference with the completion circuits $P_{\mathcal{Q}_v,e}\cup \{e\}$ corresponding to its child boundary elements $e\in S_v$, i.e., for every $i\in [2]$,
    \[
        F_i^+=\bigcup_{e\in F_i}A_e=F_i \SymDif \left(\SymDif_{v \text{ is a child of } u}\left(\SymDif_{e \in F_i\cap S_v}\left(P_{\mathcal{Q}_v,e}\cup \{e\}\right)\right)\right).
    \]
    Thus, in $F_1^+ \SymDif F_2^+$, completion circuits corresponding to boundary elements shared by the two paths occur twice and cancel. Consequently,
    \[
        F_1^+\SymDif F_2^+ = (F_1 \SymDif F_2) \SymDif \left(\SymDif_{v \text{ is a child of } u} \left(\SymDif_{e\in (F_1\SymDif F_2)\cap S_v} \left(P_{\mathcal{Q}_v,e}\cup \{e\}\right)\right)\right).
    \]

    For every child $v$ of $u$, we define
    \[
        B_v:=\SymDif_{e\in (F_1\SymDif F_2)\cap S_v} \left(P_{\mathcal{Q}_v,e}\cup \{e\}\right),
    \]
    which is a disjoint union of circuits of $M_{\mathcal{Q}_v}$. We observe that $B_v\cap S_v = (F_1\SymDif F_2)\cap S_v$. Thus, applying \Cref{lemma:sum-circuit} successively as we combine $M_u$ with its child components shows that $F_1^+\SymDif F_2^+$ is a disjoint union of circuits of $M_{\mathcal Q_u}$. Since this set avoids the parent boundary $S_u$, applying \Cref{lemma:sum-circuit} again with the empty set outside $\mathcal{Q}_u$ shows that $F_1^+ \SymDif F_2^+$ is a disjoint union of circuits of $M$.

    We next prove that this set is nonempty. Suppose that $F_1^+=F_2^+$. According to the definition of expansions, the two paths $F_1$ and $F_2$ must use the same edges in $E_u \cap E$. Moreover, they must meet the same child boundary sets, since every completion is nonempty and completions from distinct children lie in disjoint subsets of $E$. Since $F_1$ and $F_2$ are distinct, there exists a child $v$ at which the two paths use distinct boundary elements, i.e.,
    \[
        F_1 \cap S_v =\{e_1\}, \qquad F_2 \cap S_v=\{e_2\}, \qquad \text{and} \qquad e_1\neq e_2.
    \]
    Since $F_1^+=F_2^+$, we have $P_{\mathcal{Q}_v,e_1}=P_{\mathcal{Q}_v,e_2}$. Consequently,
    \[
        \left(P_{\mathcal{Q}_v,e_1}\cup \{e_1\}\right) \SymDif \left(P_{\mathcal{Q}_v,e_2}\cup \{e_2\}\right) = \{e_1,e_2\}
    \]
    is a disjoint union of circuits of the matroid $M_{\mathcal{Q}_v}$. This contradicts \Cref{remark:subtree-sums}, since $\{e_1,e_2\}\subsetneq S_v$ contains a circuit of $M_{\mathcal{Q}_v}$.
    
    Hence, $F_1^+\SymDif F_2^+$ is a nonempty disjoint union of circuits of $M$. It follows that
    \[
        g \leq |F_1^+\SymDif F_2^+| \leq |F_1^+|+|F_2^+| = \ell_u(F_1)+\ell_u(F_2) < g,
    \]
    which leads to a contradiction. Therefore, there exists at most one admissible simple $x$-$y$ path of length less than $g/2$.
\end{proof}

We now construct certificates at node $u$ as follows. For a first pass, we encourage readers to consider the special case where $u$ is the root node of $\mathcal{T}$ and has no children.

\begin{lemma}\label{lemma:graphic-certificates}
    Let $u\in V(\mathcal{T})$ be a graphic node and $m_u:=|E_u|$. Then, there exists a collection $\mathcal{W}_u$ of at most $3m_u^2$ sets satisfying:
    \begin{enumerate}
        \item Every set $W\in \mathcal{W}_u$ is a subset of $E_{\mathcal{Q}_u}\setminus S_u$ with $|W|\geq g/2$. 
        \item For every subset $R\subseteq E$, if $u$ is not $R$-good and all its children are $R$-good, then there exists $W\in \mathcal{W}_u$ with $W\subseteq R$.
    \end{enumerate}
    Furthermore, if $R$ is the remaining set returned by \Cref{algo:random-delete}, then 
    \[
        \Pr[\exists\ W\in \mathcal{W}_u \text{ with } W\subseteq R] \leq \frac{3m_u^2}{m^{20}}.
    \]
\end{lemma}
\begin{proof}
    We first construct the collection $\mathcal{W}_u$ of sets as follows:
    \begin{enumerate}
        \item For every edge $e\in E(H_u)$ with length $\ell_u(e)\geq g/2$, we add $A_e$ to $\mathcal{W}_u$.
        \item For every admissible simple path $F$ of $H_u$ with at least two edges and length $\ell_u(F)\geq g/2$, if deleting its first edge and deleting its last edge both make its length strictly less than $g/2$, then we add its expansion $F^+$ to $\mathcal{W}_u$.
    \end{enumerate}

    We observe that every set $W\in \mathcal{W}_u$ has $W\subseteq E_{\mathcal{Q}_u}\setminus S_u$ since $H_u$ is obtained from $G_u$ by deleting the parent boundary edges $S_u$. Moreover, we have $|A_e|=\ell_u(e)$ for every edge $e\in E(H_u)$ and $|F^+|=\ell_u(F)$ for every admissible path $F$ of $H_u$. Thus, every set $W\in \mathcal{W}_u$ has $|W|\geq g/2$.

    We now bound the number of sets in $\mathcal{W}_u$. There are at most $m_u$ sets of the first type. For the second type, we consider both orientations of every qualifying path. We encode each oriented path $F$ by the ordered pair consisting of its first and last edges, with both edges directed according to the orientation of $F$. Let $x$ and $y$ be the initial endpoints of the first edge and the last edge, respectively. We observe that deleting the last edge from $F$ leaves an admissible simple $x$-$y$ path of length less than $g/2$. Since $F$ has at least two edges, $x\neq y$. By \Cref{lemma:graphic-short-path}, this shorter path is uniquely determined. Appending the recorded last edge then recovers $F$. Therefore, the encoding uniquely determines the oriented path. Since there are $2m_u$ directed edges, the number of possible encodings is bounded by $4m_u^2$. We observe that each undirected path has two orientations, which implies that there are at most $2m_u^2$ paths of the second type. Therefore,
    \[
        |\mathcal{W}_u|\leq m_u+2m_u^2\leq 3m_u^2.
    \]

    We next show that $\mathcal{W}_u$ satisfies the second property. We fix an arbitrary set $R\subseteq E$ such that $u$ is not $R$-good and all its children are $R$-good. According to \Cref{lemma:target-free-absorp}, there exist corresponding circuits $D$ of $M_u$ and $\widehat{D}$ of $\widehat{M_u}$ such that
    \[
        |D\cap S_u|\leq 1,
        \qquad
        D^+\setminus S_u\subseteq R,
        \qquad\text{and} \qquad
        w_u(\widehat{D})\geq g/2.
    \]
    Moreover, $w_u(\widehat{D})\geq g$ when $D\cap S_u=\emptyset$. Let $F:=D\setminus S_u$, which meets every child boundary in at most one element and is therefore admissible. The parent boundary elements $S_u$ remain unchanged under expansion, which implies that
    \[
        F^+=D^+\setminus S_u \subseteq R.
    \]
    Since $F$ is admissible, its expansion has size $|F^+|=\ell_u(F)$. Furthermore, the weight function established in \Cref{section:target-free-absorbing}, with the parent boundary weights set to zero, gives $w_u(\widehat{D})=|D^+ \setminus S_u|$. Consequently,
    \[
        \ell_u(F)=|F^+|=|D^+ \setminus S_u|=w_u(\widehat{D}).
    \]
    We recall that $D$ is a cycle in $G_u$ and $F=D\setminus S_u$. This implies that $F$ is an admissible simple path of length at least $g/2$ when $|D\cap S_u|=1$, and an admissible cycle of length at least $g$ otherwise.

    If $F$ contains an edge $e$ with $\ell_u(e)\geq g/2$, then $A_e\in \mathcal{W}_u$ and $A_e\subseteq R$. Thus, we may assume that every edge of $F$ has length less than $g/2$. If $F$ is a cycle, deleting an arbitrary edge leaves a simple path of length greater than $g/2$. Hence, in either case, we obtain an admissible simple path of length at least $g/2$ whose expansion is contained in $R$. We repeatedly delete the first edge or the last edge whenever the remaining path still has length at least $g/2$, and let $F'\subseteq F$ be the path obtained when no further deletion is possible. Then, $\ell_u(F')\geq g/2$ and deleting the first edge and deleting the last edge both make its length less than $g/2$. Moreover, $F'$ contains at least two edges since every edge has length less than $g/2$. This implies that its expansion $(F')^+\in \mathcal{W}_u$ and $(F')^+\subseteq R$.

    Finally, suppose that $R$ is the remaining set returned by \Cref{algo:random-delete}. Every set $W\in \mathcal{W}_u$ has size at least $g/2$, which implies that
    \[
        \Pr[W\subseteq R] = m^{-40|W|/g} \leq m^{-20}.
    \]
    Thus, a union bound over the sets in $\mathcal{W}_u$ yields
    \[
        \Pr[\exists\ W\in \mathcal{W}_u \text{ with } W\subseteq R] \leq m^{-20}\cdot |\mathcal{W}_u| \leq \frac{3m_u^2}{m^{20}}.
    \]
\end{proof}

\subsection{Certificates at Cographic Nodes}\label{section:target-free-cographic}

In this section, we construct certificates at cographic nodes and bound their survival probabilities. Let $u\in V(\mathcal{T})$ be a cographic node that is not graphic. Let $M_u=M^*(G_u)$ be the cographic matroid associated with $u$ with $G_u=(V_u,E_u)$ being the underlying graph. We recall that for every child $v$ attached through a $2$- or $3$-sum and every element $s\in S_v$ in the corresponding boundary set, $P_{\mathcal{Q}_v,s}$ is the completion of $s$ in $\mathcal{Q}_v$ with minimum cardinality. Let $S_u$ be the parent boundary set.

As discussed in \Cref{section:target-free-absorbing}, absorbing the child components produces a weighted cographic matroid $\widehat{M_u}=M^*(\widehat{G_u})$, where the parent boundary edges $S_u$ are preserved and have weight zero. We emphasize that graphs $G_u$ and $\widehat{G_u}$ may have parallel edges. For every circuit $\widehat{D}$ of $\widehat{M_u}$, let $D$ be its corresponding circuit of $M_u$ and let $D^+$ be its expansion through the child components. We recall that
\[
    \widehat{D} \cap S_u = D\cap S_u \qquad \text{and} \qquad w_u(\widehat{D})=|D^+\setminus S_u|.
\]
We first establish a bound on the number of circuits in the lemma below. The lemma is trivial if $u$ is the root node of $\mathcal{T}$ and has no children, and it should be seen as a generalization of this special case.

\begin{lemma}\label{lemma:cographic-circuit-counting}
    Let $u\in V(\mathcal{T})$ be a cographic node that is not graphic and $m_u:=|E_u|$. Then, for every $\alpha\geq 1/2$, the number of circuits $\widehat{D}$ of $\widehat{M_u}$ satisfying
    \[
        |\widehat{D} \cap S_u|\leq 1 \qquad \text{and} \qquad w_u(\widehat{D})\leq \alpha g
    \]
    is at most $5(2m_u)^{4\alpha}$.
\end{lemma}
\begin{proof}
    We recall that $\widehat{M_u}$ is a cographic matroid and every circuit $\widehat{D}$ is a nonempty bond (minimal cut) in the underlying graph $\widehat{G_u}$. We first consider cuts that avoid the parent boundary. Let $H_u$ be the graph obtained from $\widehat{G_u}$ by contracting all edges in $S_u$, with all other edge weights unchanged. We observe that a cut avoids $S_u$ precisely when the endpoints of every edge in $S_u$ lie on the same side, which implies that the cuts of $H_u$ are exactly the cuts of $\widehat{G_u}$ avoiding $S_u$.

    We now claim that every nonempty cut of $H_u$ has weight at least $g$. Let $C$ be an arbitrary nonempty cut of $H_u$. We note that there is a bond $\widehat{D}\subseteq C$ of $\widehat{G_u}$ avoiding $S_u$, which is also a circuit of matroid $\widehat{M_u}$. Let $D$ be the corresponding circuit of $M_u$ and its expansion $D^+$ is a disjoint union of circuits of $M_{\mathcal{Q}_u}$. Moreover, $D^+$ in nonempty since $D$ and every chosen completion are nonempty. Since $D^+\cap S_u=\emptyset$, applying \Cref{lemma:sum-circuit} with the empty set outside $\mathcal{Q}_u$ shows that $D^+$ contains a circuit of $M$. Consequently,
    \[
        w_u(C)\geq w_u(\widehat{D})=|D^+|\geq g.
    \]

    We next bound the number of nonempty cuts avoiding the parent boundary $S_u$. If $H_u$ contains no edge, then the number of nonempty cuts is zero. Otherwise, we may choose one vertex from each connected component of $H_u$ and merge them into one single vertex. This produces a connected graph with the same cut edge sets and weights. The resulting graph has at most $m_u$ edges and hence at most $2m_u$ vertices. As established above, its minimum cut weight is at least $g$. Thus, according to \Cref{theorem:counting-cut}, the number of nonempty cuts avoiding $S_u$ with weight at most $\alpha g$ is at most $(2m_u)^{2\alpha}$ (for $\alpha<1$ there are no such cuts, since every nonempty cut of $H_u$ has weight at least $g$).

    We next consider circuits meeting $S_u$ in exactly one element. We fix an element $s\in S_u$ and let $C_1, \ldots, C_L$ be the distinct nonempty cuts of $\widehat{G_u}$ satisfying $C_i\cap S_u = \{s\}$ and $w_u(C_i)\leq \alpha g$. We claim that $L\leq 1+(2m_u)^{4\alpha}$. This is because
    \[
        C_1 \SymDif C_2, C_1 \SymDif C_3, \ldots, C_1 \SymDif C_L
    \]
    are $L-1$ distinct nonempty cuts of $\widehat{G_u}$ avoiding $S_u$, each with weight at most $2\alpha g$. They therefore correspond to cuts of $H_u$. Applying the preceding counting bound gives
    \[
        L-1 \leq (2m_u)^{4\alpha}.
    \]

    In conclusion, the number of circuits $\widehat{D}$ of $\widehat{M_u}$ satisfying $|\widehat{D}\cap S_u|\leq 1$ and $w_u(\widehat{D})\leq \alpha g$ is at most
    \[
        (2m_u)^{2\alpha}+3\left(1+(2m_u)^{4\alpha}\right) \leq 5(2m_u)^{4\alpha},
    \]
    where the last inequality follows from $m_u\geq 1$ and $\alpha\geq 1/2$.
\end{proof}

We now construct certificates at node $u$ as follows. For a first pass, we encourage readers to consider the special case where $u$ is the root node of $\mathcal{T}$ and has no children.

\begin{lemma}\label{lemma:cographic-certificates}
    Let $u\in V(\mathcal{T})$ be a cographic node that is not graphic and $m_u:=|E_u|$. Then, there exists a collection $\mathcal{W}_u$ of sets satisfying:
    \begin{enumerate}
        \item Every set $W\in \mathcal{W}_u$ is a subset of $E_{\mathcal{Q}_u}\setminus S_u$ with $|W|\geq g/2$.
        \item For every subset $R\subseteq E$, if $u$ is not $R$-good and all its children are $R$-good, then there exists $W\in \mathcal{W}_u$ with $W\subseteq R$.
    \end{enumerate}
    Moreover, if $R$ is the remaining set returned by \Cref{algo:random-delete}, then 
    \[
        \Pr[\exists\ W\in \mathcal{W}_u \text{ with } W\subseteq R] \leq \frac{160m_u^4}{m^{20}}.
    \]
\end{lemma}
\begin{proof}
    We first construct the collection $\mathcal{W}_u$ of sets as follows. For every circuit $\widehat{D}$ of $\widehat{M_u}$ satisfying $|\widehat{D}\cap S_u|\leq 1$ and $w_u(\widehat{D})\geq g/2$, we let $D$ be its corresponding circuit of $M_u$ and add $D^+\setminus S_u$ to $\mathcal{W}_u$. We observe that every set $W\in \mathcal{W}_u$ satisfies $W\subseteq E_{\mathcal{Q}_u}\setminus S_u$. Moreover, the weight function established in \Cref{section:target-free-absorbing} gives
    \[
        |W|=|D^+\setminus S_u|=w_u(\widehat{D})\geq g/2.
    \]

    We next show that $\mathcal{W}_u$ satisfies the second property. We fix an arbitrary set $R\subseteq E$ such that $u$ is not $R$-good and all its children are $R$-good. According to \Cref{lemma:target-free-absorp}, there exist corresponding circuits $D$ of $M_u$ and $\widehat{D}$ of $\widehat{M_u}$ such that
    \[
        |D\cap S_u|\leq 1,
        \qquad
        D^+\setminus S_u\subseteq R,
        \qquad\text{and} \qquad
        w_u(\widehat{D})\geq g/2.
    \]
    Since $\widehat{D}\cap S_u=D\cap S_u$, the circuit $\widehat{D}$ satisfies the conditions in our construction and hence, $D^+\setminus S_u\in \mathcal{W}_u$ and $D^+\setminus S_u\subseteq R$.

    Finally, suppose that $R$ is the remaining set returned by \Cref{algo:random-delete}. Based on \Cref{lemma:cographic-circuit-counting}, for every $\alpha\geq 1/2$, the number of sets in $\mathcal{W}_u$ of size at most $\alpha g$ is at most $5(2m_u)^{4\alpha}$. We recall that for every set $W\in \mathcal{W}_u$, its survival probability is
    \[
        \Pr[W\subseteq R]=m^{-40|W|/g}.
    \]
    Therefore, applying a union bound shows that
    $$\begin{aligned}
        \Pr[\exists\ W\in \mathcal{W}_u \text{ with } W\subseteq R] & \leq \sum_{W\in \mathcal{W}_u} m^{-40|W|/g} \\
        &\leq \sum_{\ell=1}^{+\infty} \left( \sum_{\substack{W\in \mathcal{W}_u \\ \frac{|W|}{g}\in \big[\frac{\ell}{2},\frac{\ell+1}{2}\big)}} m^{-40|W|/g} \right) \\
        &\leq \sum_{\ell=1}^{+\infty} m^{-20\ell}\left|\left\{W\in \mathcal{W}_u:\frac{|W|}{g}\in \big[\frac{\ell}{2},\frac{\ell+1}{2}\big)\right\}\right| \\
        &\leq \sum_{\ell=1}^{+\infty} \left(m^{-20\ell}\right) \left(5(2m_u)^{2\ell+2}\right) \\
        &= \frac{80 m_u^4}{m^{20}-4m_u^2} \leq \frac{160m_u^4}{m^{20}},
    \end{aligned}$$
    where the last inequality is because $m^{20}\geq 8m_u^2$ since $m\geq 2$ and $m\geq m_u$. The bound $m_u\leq m$ holds for every node $u\in V(\mathcal{T})$: each component of $\mathcal{T}-u$ attached to $u$ through a boundary set $S$ contains at least $|S|$ elements of $E$, by the size conditions on $2$- and $3$-sums.
\end{proof}

\subsection{Proof of \Cref{theorem:target-free}}\label{section:target-free-proof}

In this section, we complete the proof of \Cref{theorem:target-free} by constructing certificates at $R_{10}$ nodes and combining the survival bounds over the rooted decomposition tree.

\begin{proof}[Proof of \Cref{theorem:target-free}]
    Let $R$ be the remaining set returned by \Cref{algo:random-delete}.
    For the case of $2\leq m \leq 6$, there are at most $2^m$ circuits of $M$, and each circuit has at least $g$ elements and therefore survives with probability at most $m^{-40}$. Thus,
    \[
        \Pr[R \not\in \mathcal{I}] \leq 2^m m^{-40}\leq \frac{1}{m^5}.
    \]
    We henceforth assume that $m\geq 7$.

    We first construct a collection of sets at $R_{10}$ nodes. Let $u \in V(\mathcal{T})$ be a node whose matroid $M_u$ is isomorphic to $R_{10}$. Since $R_{10}$ contains no circuit of size three, every boundary set incident to $u$ has size at most one. Thus, every circuit of $M_u$ meets every child boundary in at most one element and admits an expansion through the child components. We define a collection $\mathcal{W}_u$ of sets as follows:
    \[
        \mathcal{W}_u:=\left\{D^+\setminus S_u: D \text{ is a circuit of } M_u, \ |D^+\setminus S_u|\geq g/2\right\}.
    \]
    We observe that every set $W\in \mathcal{W}_u$ satisfies $W\subseteq E_{\mathcal{Q}_u}\setminus S_u$ and $|W|\geq g/2$. Moreover, based on \Cref{lemma:target-free-absorp}, for a subset $X\subseteq E$, if $u$ is not $X$-good and all its children are $X$-good, then there exists $W\in\mathcal{W}_u$ with $W\subseteq X$. Since $R_{10}$ has ten elements, it has at most $2^{10}=1024$ circuits. Therefore,
    \[
        \Pr[\exists\ W\in \mathcal{W}_u \text{ with } W\subseteq R] \leq \frac{1024}{m^{20}}.
    \]

    Together with \Cref{lemma:graphic-certificates,lemma:cographic-certificates}, we have constructed a collection $\mathcal{W}_u$ of sets at every node $u$ such that
    \[
        \Pr[\exists\ W\in \mathcal{W}_u \text{ with } W\subseteq R] \leq \frac{160m_u^4}{m^{20}},
    \]
    where $m_u=|E_u|$. Suppose that no set in any collection $\mathcal{W}_u$ is contained in $R$. By \Cref{lemma:graphic-certificates,lemma:cographic-certificates} and the preceding argument for $R_{10}$ nodes, every node whose children are all $R$-good is itself $R$-good. In particular, every leaf is $R$-good, since it has no children. Induction from the leaves to the root therefore shows that every node is $R$-good. Thus, the root is $R$-good and $R$ is independent.

    Finally, we bound the probability that $R$ is dependent. By \Cref{lemma:unordered-decomposition-tree-size}, $|V(\mathcal{T})|\leq 2m$. Since $m_u\leq m$ for every node $u\in V(\mathcal{T})$, applying a union bound gives
    $$\begin{aligned}
        \Pr[R\not\in \mathcal{I}] &\leq \sum_{u\in V(\mathcal{T})}\Pr[\exists\ W\in \mathcal{W}_u \text{ with } W\subseteq R] \\
        &\leq \sum_{u\in V(\mathcal{T})} \frac{160m_u^4}{m^{20}} \leq \frac{320m^5}{m^{20}} \leq \frac{1}{m^5},
    \end{aligned}$$
    where the last inequality is because $m\geq 7$.
\end{proof}

\section{Conditional Nullity and Circuit Counting}\label{section:new:conditional-nullity}

In this section, we prove \Cref{theorem:nullity} and derive the counting and enumeration corollaries. In \Cref{subsection:regular-erdos-posa}, we state a bound relating the minimum circuit size of a regular matroid to its nullity, which generalizes a theorem of Erd\H{o}s and P\'osa that relates girth to cyclomatic number in graphs. We then use this bound to analyze nullity conditioned on circuit survival in \Cref{subsection:conditional-nullity}. Finally, \Cref{subsection:counting-enumerating} establishes the counting and enumeration corollaries.

\subsection{Minimum Circuit Bound for Regular Matroids}\label{subsection:regular-erdos-posa}
Erd\H{o}s and P\'osa~\cite{erdHos1962maximal} showed that a graph with $m$ edges and cyclomatic number $s\geq 1$ contains a cycle of length $O(m\log(s+1)/s)$. 
Knauer's recent work \cite{Kn26} extends this asymptotic bound to regular matroids as stated below. 

\begin{theorem}\label{theorem:regular-girth-nullity}
    Let $M=(E,\mathcal{I})$ be a regular matroid with $m$ elements, minimum circuit size $g$, and nullity $s:=\nu_M(E)\geq 1$. Then,
    \[
        g \leq \frac{168m\ln(s+1)}{s}.
    \]
\end{theorem}

The Erd\H{o}s-P\'osa bound is essentially tight via the existence of cubic graphs with logarithmic girth \cite{ErdosSachs1963}. 
Although the above-mentioned generalization to regular matroids is known via \cite{Kn26}, we give an alternate proof in this section using Theorem \ref{theorem:target-free}. 
We first derive a weaker bound $g=O(m\log m/s)$ from \Cref{theorem:target-free}. Then, we construct a smaller regular matroid with the same nullity and only $O(s^2)$ elements, without decreasing the ratio of minimum circuit size to ground set size by more than a constant factor. Applying the weaker bound to this smaller matroid replaces $\log m$ with $O(\log(s+1))$ and yields the desired bound. We formalize the weaker bound as follows.

\begin{lemma}\label{lemma:girth-from-sampling}
    Let $M=(E,\mathcal{I})$ be a regular matroid with $m\geq 2$ elements, minimum circuit size $g$, and nullity $s:=\nu_M(E)\geq 1$. Then,
    \[
        g \leq \frac{42m\ln m}{s}.
    \]
\end{lemma}
\begin{proof}
    We apply \Cref{algo:random-delete} to the regular matroid $M$, and let $d$ be a random variable indicating the number of deleted elements. Whenever the remaining set is independent, at least $s$ elements have been deleted, which implies that $d\geq s$. According to \Cref{theorem:target-free}, the remaining set is independent with probability at least $(1-1/m^5)$. Hence,
    \[
        \mathbb{E}[d]\geq s(1-m^{-5})\geq \frac{31}{32}s,
    \]
    where the last inequality is because $m\geq 2$. On the other hand, every element is deleted with probability $1-m^{-40/g}$. Therefore,
    \[
        \mathbb{E}[d]=m\left(1-m^{-40/g}\right) \leq \frac{40m\ln m}{g},
    \]
    where the inequality follows from $1-\exp(-x)\leq x$ for every $x\geq 0$. Combining these two bounds gives
    \[
        g\leq \frac{40m\ln m}{\mathbb{E}[d]}\leq \frac{40m\ln m}{\frac{31}{32}s}\leq \frac{42m\ln m}{s}.
    \]
\end{proof}

We next use Heller's theorem to construct a smaller regular matroid with the same nullity. We first introduce the necessary terminology. An element $e\in E$ is a \emph{loop} if $\{e\}$ is a circuit. Two distinct elements $e, f\in E$ are \emph{parallel} if $\{e,f\}$ is a circuit. The non-loop elements of a matroid can be partitioned into \emph{parallel classes}, each of which is an inclusion-wise maximal set whose distinct elements are pairwise parallel; singleton classes are allowed. A matroid is \emph{simple} if it has no loops and no pair of parallel elements. Heller's theorem relates the rank and the size of the ground set in simple regular matroids as follows.

\begin{theorem}[Heller~\cite{heller1957linear}]\label{theorem:heller}
    Let $M=(E,\mathcal{I})$ be a simple regular matroid of rank $r\geq 1$. Then,
    \[
        |E|\leq \frac{r(r+1)}{2}.
    \]
\end{theorem}

We now apply Heller's theorem to the dual matroid, whose rank equals the nullity of the original matroid. By adjusting the sizes of its parallel classes, we obtain the following reduction.

\begin{lemma}\label{lemma:regular-nullity-compression}
    Let $M=(E,\mathcal{I})$ be a regular matroid with $m:=|E|$ elements, minimum circuit size $g$, and nullity $s:=\nu_M(E)\geq 1$. There exists a regular matroid $\tilde{M}=(\tilde{E},\tilde{\mathcal{I}})$ with $\tilde{m}:=|\tilde{E}|$ elements, minimum circuit size $\tilde{g}$, and nullity $\tilde{s}$ such that
    \begin{enumerate}
        \item $\tilde{s}=s$;
        \item $2\leq \tilde{m}\leq s(s+1)$; and
        \item $\tilde{g}/\tilde{m} \geq g/(2m)$.
    \end{enumerate}
\end{lemma}
\begin{proof}
    Since $M$ is a regular matroid, its dual matroid $M^*$ is also regular and has rank
    \[
        r(M^*)=m-r(M)=\nu_M(E)=s.
    \]
    Let $P_1, \ldots, P_k$ be the parallel classes of $M^*$. For every $i\in [k]$, we define
    \[
        w_i:=|P_i| \qquad \text{and} \qquad W:=\sum_{i=1}^{k}w_i\leq m.
    \]
    Now, we delete all loops of $M^*$ and retain one representative $e_i$ from the parallel class $P_i$ for every $i\in [k]$. Let $M_0^*$ denote the resulting matroid. It is simple and regular, and its rank remains $s$. Thus, according to \Cref{theorem:heller},
    \begin{align}
        s\leq k\leq \frac{s(s+1)}{2}. \label{inequality:heller}
    \end{align}

    We next describe the cocircuits of $M^*$ in terms of these parallel classes. We recall that a cocircuit is an inclusion-wise minimal set whose deletion decreases the rank. We note that deleting loops or removing some but not all elements of a parallel class does not affect the rank. Moreover, deleting entire parallel classes decreases the rank of $M^*$ precisely when deleting their representatives decreases the rank of $M_0^*$. Consequently, the cocircuits of $M^*$ are exactly
    \[
        \left\{ \left(\bigcup_{i:e_i\in D}P_i\right) : D \text{ is a cocircuit of } M_0^* \right\}.
    \]
    Since the cocircuits of $M^*$ are the circuits of $M$, for every cocircuit $D$ of $M_0^*$, we have
    \[
        \sum_{i:e_i\in D}w_i\geq g.
    \]

    We now construct the matroid $\tilde{M}$. For every $i\in [k]$, we replace each representative $e_i$ in $M_0^*$ with a parallel class $\tilde{P_i}$ of size
    \[
        a_i:=\left\lceil \frac{(k+1)w_i}{W}\right \rceil\geq 1,
    \]
    and let $\tilde{M}$ be the dual of the resulting matroid. Explicitly, a set is independent in $\tilde{M}^*$ if it contains at most one element from each class $\tilde{P_i}$ and the corresponding representatives form an independent set of $M_0^*$. Replacing elements with parallel copies preserves regularity and rank. Therefore, $\tilde{M}^*$ is regular and has rank $s$, which implies that $\tilde{M}$ is also regular and has nullity $\tilde{s}=s$.

    We next bound the number of elements in $\tilde{M}$. According to the definition of $a_i$, we have
    \[
        \tilde{m}=\sum_{i=1}^{k}a_i=\sum_{i=1}^{k}\left\lceil \frac{(k+1)w_i}{W}\right\rceil\geq \sum_{i=1}^{k} \frac{(k+1)w_i}{W}=k+1
    \]
    and
    \[
        \tilde{m}=\sum_{i=1}^{k}a_i=\sum_{i=1}^{k}\left\lceil \frac{(k+1)w_i}{W}\right\rceil< \sum_{i=1}^{k} \left(\frac{(k+1)w_i}{W}+1\right)=2k+1.
    \]
    Since $\tilde{m}$ is an integer, $\tilde{m}\leq 2k$. Together with inequality~(\ref{inequality:heller}), we have $2\leq \tilde{m}\leq s(s+1)$.

    Finally, we bound the minimum circuit size of $\tilde{M}$. We observe that the correspondence between cocircuits and parallel classes established above also applies to $\tilde{M}^*$. More precisely, the cocircuits of $\tilde{M}^*$ are exactly
    \[
        \left\{ \left(\bigcup_{i:e_i\in D}\tilde{P_i}\right) : D \text{ is a cocircuit of } M_0^* \right\}.
    \]
    Therefore, every circuit $C$ of $\tilde{M}$ corresponds to a cocircuit $D$ of $M_0^*$. It implies that
    \[
        |C|=\sum_{i:e_i\in D}a_i=\sum_{i:e_i\in D}\left\lceil \frac{(k+1)w_i}{W}\right\rceil\geq \frac{(k+1)}{W}\left(\sum_{i:e_i\in D}w_i\right)\geq \frac{(k+1)g}{W}.
    \]
    Therefore, together with $W\leq m$ and $\tilde{m}\leq 2k$, we have
    \[
        \tilde{g}\geq \frac{(k+1)g}{W} \geq \frac{(k+1)g}{m} \geq \frac{\tilde{m}}{2m}g,
    \]
    which implies that $\tilde{g}/\tilde{m} \geq g/(2m)$.
\end{proof}

We now combine the preceding two lemmas to complete the proof of \Cref{theorem:regular-girth-nullity}.

\begin{proof}[Proof of \Cref{theorem:regular-girth-nullity}]
    Let $\tilde{M}$ be the regular matroid in \Cref{lemma:regular-nullity-compression} with $\tilde{m}$ elements, minimum circuit size $\tilde{g}$, and nullity $\tilde{s}$. According to \Cref{lemma:regular-nullity-compression}, we have
    $$\begin{aligned}
        \frac{g}{2m}&\leq \frac{\tilde{g}}{\tilde{m}} \leq \frac{42\ln \tilde{m}}{\tilde{s}} \ \ \  \ \text{(by applying \Cref{lemma:girth-from-sampling} on matroid $\tilde{M}$)} \\
        &\leq  \frac{84\ln (s+1)}{s}, \ \ \ \ \text{(since $\tilde{s}=s$ and $\tilde{m}\leq s(s+1)$)}
    \end{aligned}$$
    which further implies that
    \[
        g\leq \frac{168m\ln(s+1)}{s}.
    \]
\end{proof}

\subsection{Nullity conditioned on Circuit Survival}\label{subsection:conditional-nullity}

In this section, we prove the conditional nullity bound in \Cref{theorem:nullity}. 
We first use \Cref{theorem:regular-girth-nullity} to infer the following Corollary.

\begin{corollary}\label{coro:low-girth}
    Let $M=(E, \mathcal{I})$ be a regular matroid, $g$ be the minimum circuit size of $M$, $C$ be a circuit in $M$ of size at most $\alpha g$ and $X\subseteq E\setminus C$ be a set disjoint from $C$. Let $M_X$ be the matroid obtained from $M$ by contracting the elements of $X$ and deleting all remaining elements outside $C$. If the nullity of $C$ in $M_X$ is at least $\lceil 2^{20}\alpha \ln(\alpha+2)\rceil$, then $M_X$ contains a circuit of size at most $g/40$. 
\end{corollary}
\begin{proof}
    Let $L:=\lceil 2^{20}\alpha \ln(\alpha+2)\rceil$. \Cref{theorem:regular-girth-nullity} shows that there exists a circuit $C_X\subseteq C$ of $M_X$ satisfying
    $$\begin{aligned}
        |C_X| &\leq \frac{168|C|\ln (\nu_{M_X}(C)+1)}{\nu_{M_X}(C)} \\
        &\leq \frac{168|C|\ln(L+1)}{L} \ \ \ \ \text{(since $\nu_{M_X}(C)>L$ and $\frac{\ln(x+1)}{x}$ is decreasing for $x>0$)} \\
        &\leq 168\alpha g \left( \frac{\ln(L+1)}{L} \right) \ \ \ \ \text{(since $|C|\leq \alpha g$)} \\
        &\leq 168\alpha g \left( \frac{\ln(2^{21}\alpha \ln(\alpha+2))}{2^{20} \alpha \ln(\alpha+2)} \right) \ \ \ \ \text{(since $L=\left\lceil 2^{20} \alpha \ln(\alpha+2)\right\rceil$)} \\
        &\leq 168\alpha g \left( \frac{17\ln(\alpha+2)}{2^{20} \alpha \ln(\alpha+2)} \right) \ \ \ \ \text{(since $\ln(2^{21}\alpha\ln(\alpha+2))\leq 17\ln(\alpha+2)$ for $\alpha \geq 1$)} \\
        &\leq \frac{g}{40}.
    \end{aligned}$$
\end{proof}

We now prove Theorem \ref{theorem:nullity} using Corollary \ref{coro:low-girth}. 
\begin{proof}[Proof of \Cref{theorem:nullity}]
    We recall that $R$ is the set returned by \Cref{algo:random-delete}. Since every element $e\in E$ survives independently with probability $m^{-40/g}$ and $|C|\leq \alpha g$, we have
    \[
        \Pr[C\subseteq R]=m^{-40|C|/g}\geq m^{-40\alpha}.
    \]
    
    We next bound the nullity conditioned on $C\subseteq R$. We independently sample $X\subseteq E\setminus C$ and $Y\subseteq C$ by retaining each element with probability $m^{-40/g}$. We observe that the set $U:=X\cup Y$ has the distribution of an unconditional run of \Cref{algo:random-delete}, whereas $X\cup C$ has the distribution of $R$ conditioned on $C\subseteq R$. According to \Cref{theorem:target-free}, $\Pr[U\notin \mathcal{I}] \leq m^{-5}$. Since $X\subseteq U$, we further have
    \begin{align}
        \Pr[\nu_M(X)>0] = \Pr[X\not\in \mathcal{I}] \leq m^{-5}. \label{inequality:X-bound}
    \end{align}

    Let $M_X$ be the regular matroid on $C$ obtained by contracting $X$ in the matroid $M$ and then deleting every remaining element outside $C$. The contraction rank formula in \Cref{definition:contraction} gives that
    \begin{align}
        r_{M_X}(C)=r_M(X\cup C)-r_M(X). \label{equality:rank-contraction}
    \end{align}
    This implies that
    \begin{align}
        \nu_M(X\cup C) &= |X\cup C|-r_M(X\cup C) \ \ \text{(by definition of nullity)} \notag \\
        &= |X\cup C| - r_{M_X}(C)-r_M(X) \ \ \text{(by equality~(\ref{equality:rank-contraction}))} \notag \\
        &= \left(|X|-r_M(X)\right) + \left(|C|-r_{M_X}(C)\right) \notag \\
        &= \nu_M(X) + \nu_{M_X}(C). \ \ \text{(by definition of nullity)} \label{equality:nullity-decomposition}
    \end{align}

    By inequality~(\ref{inequality:X-bound}), $\nu_M(X)=0$ with high probability. We define
    \[
        L:=\left\lceil 2^{20} \alpha \ln(\alpha+2)\right\rceil.
    \]
    We next bound the probability that $\nu_{M_X}(C)>L$. 
    By Corollary \ref{coro:low-girth}, if $\nu_{M_X}(C)>L$, then the matroid $M_X$ contains a circuit of size at most $g/40$.
        For every sampled $X$ with $\nu_{M_X}(C)>L$, we select one such circuit $C_X$. We emphasize that the selection of $C_X$ is independent of $Y$. We recall that $M_X$ is the regular matroid on $C$ obtained by contracting $X$ in $M$ and then deleting every remaining element outside $C$, and $C_X$ is a circuit of $M_X$. Hence,
    $$\begin{aligned}
        r_M(X\cup C_X) &= r_M(X)+r_{M_X}(C_X) \ \ \text{(by \Cref{definition:contraction})} \\
        &= r_M(X)+|C_X|-1 \ \ \text{(since $C_X$ is a circuit of $M_X$)} \\
        &\leq |X|+|C_X|-1,
    \end{aligned}$$
    which implies that $X\cup C_X$ is dependent in the matroid $M$. Consequently, if $C_X\subseteq Y$, then $U=X\cup Y$ is dependent. We note that conditioned on every fixed $X$ with $\nu_{M_X}(C)>L$, the elements of $Y$ are still retained independently with probability $m^{-40/g}$. Thus, for every fixed $X$ with $\nu_{M_X}(C)>L$,
    $$\begin{aligned}
        \Pr[U\not\in \mathcal{I} \mid X] &\geq \Pr[C_X\subseteq Y \mid X] \ \ \text{(since $U$ is dependent if $C_X\subseteq Y$)} \\
        &= \left(m^{-40/g}\right)^{|C_X|} \ \ \text{(since $C_X$ is independent of $Y$)} \\
        &\geq \left(m^{-40/g}\right)^{g/40} \ \ \text{(since $|C_X|\leq g/40$)} \\
        &= m^{-1}.
    \end{aligned}$$
    We recall that $\Pr[U\not\in \mathcal{I}]\leq m^{-5}$. Moreover,
    $$\begin{aligned}
        \Pr[U\not\in \mathcal{I}] \geq \Pr[U\not\in \mathcal{I} \mid \nu_{M_X}(C)>L]\cdot \Pr[\nu_{M_X}(C)>L] \geq \left(m^{-1}\right)\Pr[\nu_{M_X}(C)>L],
    \end{aligned}$$
    which further implies that
    \begin{align}
        \Pr[\nu_{M_X}(C)>L] \leq m^{-4}. \label{inequality:Y-bound}
    \end{align}

    Therefore,
    $$\begin{aligned}
        \Pr[\nu_M(R)\leq L \mid C\subseteq R] &= \Pr[\nu_M(X\cup C)\leq L] \\
        &= \Pr[\nu_M(X)+\nu_{M_X}(C)\leq L] \ \ \text{(by equality~(\ref{equality:nullity-decomposition}))} \\
        &\geq 1 - \Pr[\nu_M(X)>0] - \Pr[\nu_{M_X}(C)>L] \ \ \text{(by union bound)} \\
        &\geq 1 - m^{-5} - m^{-4} \ \ \text{(by inequalities~(\ref{inequality:X-bound}) and (\ref{inequality:Y-bound}))} \\
        &\geq 1-m^{-2}. \ \ \text{(since $m\geq 2$)}
    \end{aligned}$$
\end{proof}

\begin{remark}
    The proof above does not require $C$ to be a circuit. Consequently, \Cref{theorem:nullity} holds for every fixed set $S\subseteq E$ with $|S|\leq \alpha g$.
\end{remark}

\subsection{Counting and Enumerating Near-Minimum Circuits}\label{subsection:counting-enumerating}

We combine \Cref{theorem:nullity} with \Cref{lemma:binary-matroid-nullity} to derive the counting and enumeration corollaries. Throughout this section, we assume $\alpha\leq m$, since otherwise every circuit is an $\alpha$-approximate minimum circuit.

\begin{proof}[Proof of \Cref{cor:counting}]
    Let $\mathcal{C}_{\alpha}$ be the collection of circuits of size at most $\alpha g$. According to \Cref{theorem:nullity}, every fixed circuit in $\mathcal{C}_{\alpha}$ survives together with a nullity bound of $O(\alpha\log(\alpha+2))$ with probability at least $m^{-O(\alpha)}$. Since $M$ is binary, \Cref{lemma:binary-matroid-nullity} implies that every sample satisfying this nullity bound contains at most $2^{O(\alpha\log(\alpha+2))}$ circuits. Counting surviving circuits only in these samples and taking expectations gives
    \[
        |\mathcal{C}_{\alpha}|m^{-O(\alpha)} \leq 2^{O(\alpha\log(\alpha+2))}=(\alpha+2)^{O(\alpha)},
    \]
    which further implies that
    \[
        |\mathcal{C}_{\alpha}|\leq(m\alpha)^{O(\alpha)}=m^{O(\alpha)}.
    \]
\end{proof}

\begin{proof}[Proof of \Cref{cor:enumeration}]
    We first suppose that the minimum circuit size $g$ is known. We repeatedly run \Cref{algo:random-delete} independently. For each remaining set $R$, we compute its nullity and discard it if the nullity exceeds the $O(\alpha\log(\alpha+2))$ bound in \Cref{theorem:nullity}. Otherwise, we enumerate all circuits contained in $R$ via \Cref{lemma:binary-matroid-nullity} and retain those of size at most $\alpha g$. We return all retained circuits, with duplicates removed.

    By \Cref{theorem:nullity}, every fixed circuit of size at most $\alpha g$ survives together with the required nullity bound with probability at least $m^{-O(\alpha)}$ in each repetition. Whenever this occurs, the circuit is returned by the enumeration step. Since \Cref{cor:counting} shows that the number of these circuits is at most $m^{O(\alpha)}$, taking $m^{O(\alpha)}$ independent repetitions with a sufficiently large exponent and applying a union bound ensures that all of these circuits are returned with high probability.

    We observe that computing nullity takes polynomially many independence queries, and \Cref{lemma:binary-matroid-nullity} bounds the enumeration time per repetition by
    \[
        2^{O(\alpha\log(\alpha+2))}\operatorname{poly}(m)=(\alpha+2)^{O(\alpha)}\operatorname{poly}(m)\leq m^{O(\alpha)}.
    \]
    Thus, the total running time is $m^{O(\alpha)}$.

    If the minimum circuit size $g$ is unknown, we guess its value by trying every integer in $[m]$ and running the procedure with each guess. We collect the circuits returned under all guesses. With high probability, the correct guess yields every circuit of size at most $\alpha g$. Since every returned set is a circuit, the smallest size among the collected circuits is then exactly $g$. If no circuit is returned, we return the empty collection. Otherwise, we retain only circuits of size at most $\alpha$ times this smallest size. Trying all guesses increases the running time by only a factor of $m$.
\end{proof}

\paragraph{Acknowledgements.} See Section \ref{sec:techniques} for AI Disclosure. We thank Rohit Gurjar for help in understanding the unordered decomposition tree in \cite{gurjar2019number}, and for suggesting the Delta-Wye fix. 

\bibliographystyle{abbrv}
\bibliography{references}

\newpage
\appendix
\section{Circuits in Binary Matroids with Bounded Nullity}\label{appendix:binary-nullity}

In this section, we present the proof of \Cref{lemma:binary-matroid-nullity}. We restate and prove it below. We note that the lemma does not hold for all matroids---e.g., it fails for the uniform matroid $U_{2,4}$, which has nullity $2$ and four circuits. 

\binarymatroidnullity*
\begin{proof}
    Let $B$ be a basis of $M$ and $D:= E\setminus B$. We observe that
    \[
        |D|=|E|-|B|=|E|-r_M(E)=\nu_M(E)\leq d.
    \]
    Moreover, the mapping
    \[
        \phi:C \longmapsto C\cap D
    \]
    is injective. Indeed, suppose that two circuits $C_1$ and $C_2$ of $M$ satisfy $C_1\cap D=C_2\cap D$. Their symmetric difference is then contained in $B$, i.e., $C_1\SymDif C_2 \subseteq B$. Since $M$ is binary, $C_1 \SymDif C_2$ is a disjoint union of circuits. Since $B$ is independent, it follows that $C_1 \SymDif C_2 =\emptyset $ and hence, $C_1=C_2$. Therefore, the number of circuits is at most the number of nonempty subsets of $D$, which is at most $2^d-1$.

    It remains to show that these circuits can be enumerated efficiently. We note that a basis $B$ can be found greedily using $O(m)$ independence queries. For every $e\in D$, let $C_e$ be the unique circuit contained in $B\cup \{e\}$. This fundamental circuit can be recovered using
    \[
        C_e = \{e\} \cup \left\{b\in B: (B\setminus \{b\})\cup \{e\} \in \mathcal{I}\right\}.
    \]
    For every subset $X\subseteq D$, we define
    \[
        S_X:=\SymDif_{e\in X}C_e.
    \]
    For every circuit $C$ of $M$, $C \SymDif S_{C \cap D}$ is a symmetric difference of a circuit and a disjoint union of circuits and hence, is contained in $B$. Since $B$ is independent, we have $C \SymDif S_{C\cap D}=\emptyset$ and $C=S_{C\cap D}$. Thus, it suffices to enumerate all nonempty subsets $X\subseteq D$ and retain $S_X$ precisely when it is a circuit. This can be checked using the independence oracle by verifying that $S_X\setminus \{e\}$ is independent for every $e\in S_X$. Therefore, the total number of independence queries is $O(2^dm)$ and the total running time is $2^d\poly(m)$.
\end{proof}

\end{document}